\documentclass[a4paper,twocolumn,11pt,accepted=2026-04-02]{quantumarticle}
\pdfoutput=1

\usepackage[utf8]{inputenc}  
\usepackage[T1]{fontenc}

\usepackage{graphicx}
\usepackage{dcolumn}
\usepackage{bm}
\usepackage{xcolor}
\usepackage{qcircuit}
\usepackage{float}

\usepackage{amsmath}
\usepackage{amssymb}
\usepackage{amsthm}
\usepackage{physics}

\usepackage{hyperref}

\newtheorem{defi}{Definition}
\newtheorem{theo}{Theorem}
\newtheorem{lem}{Lemma}
\newtheorem{obs}{Observation}

\newcommand{\rd}[1]{{#1}}
\newcommand{\bl}[1]{{#1}}

\def\ra{\rangle}
\def\la{\langle}
\def\id{\mathbb{I}}
\def\Tr{\rm Tr}
\def\wm{{\color{white}-}}

\begin{document}

\title{Quantum convolutional channels and \newline multiparameter families of 2-unitary matrices}

\author{Rafał Bistroń}
\orcid{0000-0002-0837-8644}
\author{Jakub Czartowski}
\orcid{0000-0003-4062-833X}
\affiliation{Faculty of Physics, Astronomy and Applied Computer Science, Institute of Theoretical Physics, Jagiellonian University, ul. {\L}ojasiewicza 11, 30--348 Krak\'ow, Poland}
\affiliation{Doctoral School of Exact and Natural Sciences, Jagiellonian University}
\author{Karol Życzkowski}
\orcid{0000-0002-0653-3639}
\affiliation{Faculty of Physics, Astronomy and Applied Computer Science, Institute of Theoretical Physics, Jagiellonian University, ul. {\L}ojasiewicza 11, 30--348 Krak\'ow, Poland}
\affiliation{Center for Theoretical Physics, Polish Academy of Sciences, Al. Lotników 32/46, 02-668 Warszawa, Poland}

\date{09.04.2026}


\begin{abstract}
Many alternative approaches to construct quantum channels with large entangling capacity were proposed in the past decade, resulting in multiple isolated gates. In this work, we 
put forward a novel one, inspired by convolution, which
provides greater freedom of nonlocal parameters.
Although quantum counterparts of convolution have been shown not to exist for pure states, several attempts with various degrees of rigorousness have been proposed for mixed states. In this work, we follow the approach based on coherifications of multi-stochastic operations and demonstrate a surprising connection to gates with high entangling power.
In particular, we identify conditions necessary for the convolutional channels constructed using our method to possess maximal entangling power. Furthermore, we establish new, continuous classes of bipartite 2-unitary matrices of dimension $d^2$ for $d = 7$ and $d = 9$, with $2$ and $4$ free nonlocal parameters beyond simple phasing of matrix elements, corresponding to perfect tensors of rank $4$ or 4-partite absolutely maximally entangled states.
\end{abstract}

\maketitle

\section{Introduction}

Entanglement stands as a pervasive and foundational concept within the realms of quantum mechanics and quantum information. From the inception of the field, entanglement has not only captured the imagination of researchers but also steered numerous endeavours within the discipline~\cite{horodecki^4}. A natural consequence of this exploration has been the in-depth investigation of gates capable of generating substantial entanglement, giving rise to the concept of the entangling power of gates~\cite{Zanardi2000entangling}. 

Particularly noteworthy are bipartite operations that achieve maximal entangling power, known as \textit{2-unitary gates}~\cite{permutatrion_entangling_power}. 
These gates, equivalent to \textit{perfect tensors} of order~4, and 4-partite absolutely maximally entangled (AME) states~\cite{ameMasterThesis, rajchelmieldzioc2025absolutelymaximallyentangledpure}, find applications in areas as diverse as Bernoulli circuits~\cite{dual_uni_to_bernoulli_circuits}, both classical and quantum error-correcting codes~\cite{ReedSolomon1960,scott2003multipartite}, holographic codes~\cite{latorre2015holographic, pastawski2015holographic}, quantum secret sharing~\cite{HelwigEtAl2012secretSharing}, study of entanglement dynamics in quantum circuits~\cite{bertini2019exact} and others~\cite{ameMasterThesis,harper2022perfect, PhysRevA.92.032316}.
While previous constructions, based on orthogonal Latin squares~\cite{permutatrion_entangling_power} and stabilizer states~\cite{stabilizer_high_entanglement}, have yielded isolated solutions, recent developments~\cite{new_AME_16, 36_officers_of_Karol,rather2023biunimodular} suggest the existence of 2-unitary gates, and corresponding perfect tensors of rank~$4$, beyond standard constructions, potentially forming parts of non-trivial continuous families~\cite{Suahil_invariants}. However, continuous families with amplitudes differing from the already known solutions have not been known.

To meet 
the challenge of constructing such families 
we considered a seemingly disconnected problem: generalization of convolution to the setting of quantum states~\cite{bistron2023tristochastic}.
It was noted~\cite{Lomont}, that there is no proper ”operation of convolution”, which for two arbitrary pure states produces a pure state as an outcome. 
The no-go theorem, however, does not apply to density matrices. In particular, a construction of convolution of quantum states called twirled product, was recently proposed~\cite{Aniello_2019, Aniello_2023}, while other techniques were used to generate a composition of bi-partite density matrices~\cite{Karol_cos} and convolution of quantum superoperators~\cite{Sohail_2022}. 

However, none of the aforementioned techniques, sticking to all defining properties of convolution, provide operational implementation. Thus, in this work we 
follow the approach of~\cite{bistron2023tristochastic} which generalizes the construction of convolution by abandoning the associativity, while preserving other properties, especially tristochasticity \rd{which will prove to be a key feature in our study}. The quantum channels obtained in such a way can be realized by certain well-defined parameterizable unitary matrix followed by a partial trace, as presented in Fig.~\ref{fig:enter-label}.
For the sake of simplicity, while slightly abusing the terminology, hereafter we will call this construction \textit{convolutional channel}.

\begin{figure}[t]
    \centering
    \includegraphics[width=1\linewidth]{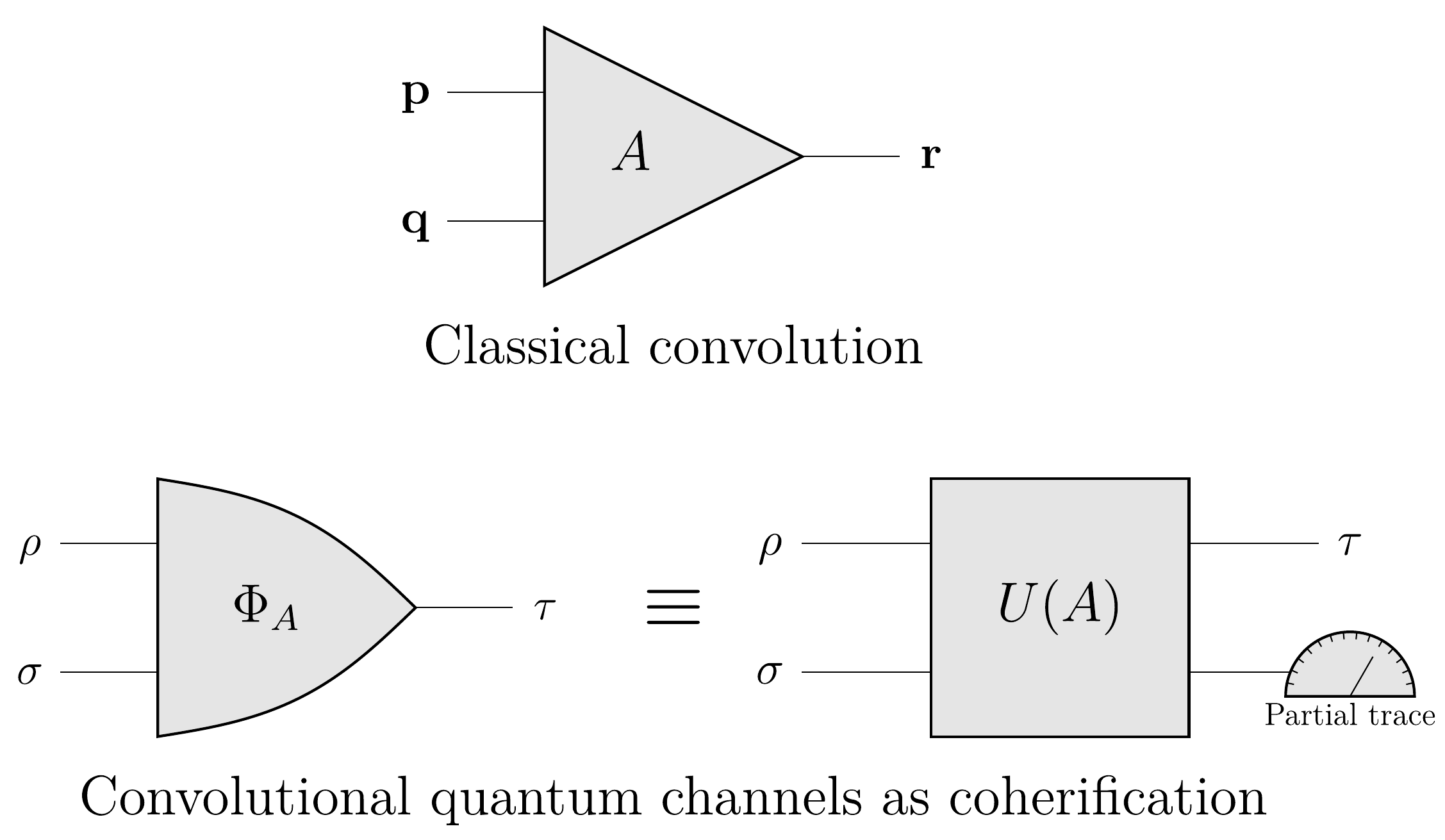}
    \caption{Classical convolution can be seen as an operation taking two probability vectors $\vb{p}$ and $\vb{q}$ as input and producing a new probability $\vb{r}$ as output. We introduce quantum convolutional channels as coherifications of tristochastic tensors $A$, which can be realized, by Stinespring representation,
    as partial traces of unitary channels $U$. }
    \label{fig:enter-label}
\end{figure}

In this work we present a construction of parametrizable 
bipartite quantum convolutional channels, based on coherifications of classical tristochastic tensors~\cite{bistron2023tristochastic}, which can be realized using bipartite unitary gates. 
First, we show
that generic channels from this family exhibit high entangling and disentangling power.
Moreover, we present necessary and sufficient conditions for our construction to provide gates with maximal \rd{values for these quantities}.
The attainability of the above conditions
is exemplified by
two novel families of bipartite 2-unitary gates
of dimension $d^2$ with $d = 7$ and $9$, parametrized by 2 and 4 nonlocal, non
trivial parameters.
\rd{Note however that presented families of 2-unitary matrices are exemplary and our framework can encompass many more classes of 2-unitary matrices.}
Furthermore, we introduce quantitative tools \rd{based on entropic coherence monotones}, which highlight differences between two locally inequivalent bipartite unitary matrices -- \rd{ranges of coherence} -- and provide estimates of these measures. Our results demonstrate that operations from our families maintain nontrivial coherence, and therefore coherence-generating abilities, under arbitrary local transformations.

The paper is organized as follows. In Section \ref{sec:sett} we invoke the concepts and notions necessary in further work, such as tristochastic tensors, entangling power and orthogonal Latin squares.
Section \ref{sec:coch_measure} introduces a new measure of coherence for unitary operators, which we later use to highlight the novelty and disparity of our construction.
Then in Section \ref{sec:tristoch} we proceed to present the entire class of convolutional channels. In Section \ref{sec:new} we present novel 2-unitary gates in dimensions $7\times 7$ and $9\times 9$, which emerged from our constructions and their parametrization goes \textit{beyond} simple phasing of matrix elements.
Finally, in the section \ref{sec:multistoch} we generalize the main concepts and results of the paper for multi-stochastic tensors and corresponding multipartite channels.
In  Section \ref{sec:outlook} we discuss obtained results and highlight important directions for further research.

The Appendix \ref{sec:2x2} is concerned with the simple example application for convectional channels as disentangling channels for an entire maximally-entangled basis.
In Appendix \ref{app:coh_measure} we describe in detail the coherence measures for unitary gates. Next, Appendix \ref{app:q_tri_call} presents calculations and proofs omitted in the main body of the paper. In Appendix \ref{app:orto6} an orthonormal matrix in dimension $6 \times 6$ with the highest known entangling power is presented.  \rd{The Appendix \ref{sec:multi_eq} discusses entangling power of multipartite unitary gates, whereas}
the Appendices \ref{app:milti_coch} and \ref{app:milti_coch_multi} serve to isolate lengthy calculations from Section \ref{sec:multistoch}. \rd{Finally in Appendix \ref{app:max_coch} we discuss the general relation between optimal coherifications of arbitrary multistochastic tensors and unitaries with maximal entangling power which goes beyond established framework based on permutation tensors.}

\textit{Relation to prior work:} Concept of coherification was introduced previously in~\cite{Korzekwa_coch}, whereas the tristochastic channels, their basic properties and connection to classical counterparts in~\cite{bistron2023tristochastic}. The concepts known from prior work are collected primarily in Section \ref{sec:sett} with a small excerpt at the beginning of Section \ref{sec:multistoch}. The remaining sections introduce novel concepts and build upon them.
Furthermore, Appendix \ref{sec:2x2}, discusses known solutions for $AME(4,4)$ from the new perspective of disentangling capabilities.

\section{Setting the scene}\label{sec:sett}

In the following subsections, we recall established notions and tools, essential to understanding the results of our paper. We start with tristochastic tensors and a method for obtaining convolutional channels by coherifying them. Then we proceed to entanglement properties of unitary operations and finish with orthogonal Latin squares and classical construction of 2-unitary gates based on pairs of such objects.

\subsection{Tristochastic tensors and coherifications}

Let us start by recalling the notions of stochasticity and tristochasticity in the classical framework which is a foundation for our work.

\begin{defi}
    A matrix $B$ is called stochastic if $B_{ij}\geq 0$ and $\sum_j B_{ij} = 1$. It is called bistochastic if $B$ and $B^T$ are both stochastic.
\end{defi}

\begin{defi}
    A tensor $A$ is called \textit{tristochastic} if $A_{ijk} \geq 0$ and $\sum_{i} A_{ijk} = \sum_{j} A_{ijk} = \sum_{k} A_{ijk} =  1$ for  any $i,j,k$.
\end{defi}

The class of tristochastic tensors, of special interest to us are permutation tensors
\rd{defined in full analogy with permutation matrices \cite{stoch_multistoch}.
\begin{defi}
    A tristochastic tensor $A$ is called a permutation tensor if its entries are restricted to zeros or ones, $A_{ijk} \in \qty{0,1}$. 
\end{defi}
As a consequence of tristochasticity, a permutation tensor $A$ has a single nonzero entry equal to one for each of its hypercolumns.}
An example of a tristochastic permutation tensor for $d = 3$ is provided below:

\begin{equation}
\label{example_permutation_tensor}
\begin{aligned}
& A = \left(\begin{matrix}
1 & 0 & 0\\
0 & 1 & 0\\
0 & 0 & 1\\
\end{matrix}\right.
\left|\begin{matrix}
0 & 1 & 0\\
0 & 0 & 1\\
1 & 0 & 0\\
\end{matrix}\right.
\left|\begin{matrix}
0 & 0 & 1\\
1 & 0 & 0\\
0 & 1 & 0\\
\end{matrix}\right)\,, \\
\end{aligned}
\end{equation}
where the reader should imagine the square sub-matrices arranged in a $3\times 3\times 3$ cube.

The action of a tristochastic tensor $A$ on a pair of probability vectors $p$, $q$ is defined analogically to the action of a stochastic matrix on 
the outer product $p\otimes q$
, i.e. $A[p,q]_i = \sum_{jk} A_{ijk} p_j q_k$. In the case when each layer of $A$ is consecutive power of permutation matrix for permutation $\sigma_i=  i+1$, as in the example \eqref{example_permutation_tensor}, the action of $A$ simplifies to
\begin{equation*}
A[p,q]_i = \sum_{jk} A_{ijk} p_j q_k = \sum_{jk} \delta_{k,i-j} p_j q_k = \sum_j p_j q_{i+j}.
\end{equation*}
\bl{This example coincides with the the ordinary convolution $[p*q]_i = \sum_j p_{j} q_{i-j}$, combined with simple preprocessing -- inverting the order of $q$ entries.} 
Thus the action of a tristochastic tensor might be interpreted as  
a generalization of the convolution of two probability vectors~\cite{bistron2023tristochastic}.

One can also define tristochasticity at the quantum level. To do so we first invoke a dynamical matrix of the channel via Choi-Jamiołkowski isomorphism~\cite{CHOI1975285, JAMIOLKOWSKI1972275}.

\begin{defi}
Let $\Omega_d$ be a set of quantum states (positive, trace one hermitian matrices) of dimension $d$.
Let $\Phi: \Omega_d \to \Omega_d$ be a quantum channel and \mbox{$|\Psi^+\ra = \sum_{i} \frac{1}{\sqrt{d}}|i\ra\otimes|i\ra$} a maximally entangled state. Then the dynamical matrix of the channel is defined as:
\begin{equation*}
D = d\cdot (\Phi \otimes \id) |\Psi^+\ra \la \Psi^+|~.
\end{equation*}
The transition from the dynamical matrix $D$ to the quantum channel is in turn defined by
\begin{equation}
\label{simple_channel}
\Phi_D(\rho) = \Phi(\rho) = \Tr_2[D (\id \otimes \rho^\top)] ,
\end{equation}
\rd{with transposition defined with respect to the basis involved in the Choi-Jamiołkowski isomorphism.}
\end{defi}
The complete positivity and trace preserving properties (CPTP) of the channel $\Phi$ are reflected by $D\geq0$ and $\Tr_1[D] = \mathbb{I}_d$, respectively.

In order to provide background, 
we start by invoking a definition of the unital channel, also known as a bistochastic channel,  in a non-standard but equivalent way:

\begin{defi}
A quantum channel $\Phi_D: \Omega_d \to \Omega_d$ defined by the dynamical matrix $D_{j_1, j_2}^{i_1,i_2}$ by
\begin{equation}
\Phi_D[\rho] = \Tr_{2}[D (\mathbb{I} \otimes \rho^\top)]~,
\end{equation}
is bistochastic if the map:
\begin{equation}
\Tr_{1}\;[D(\rho^\top \otimes \mathbb{I}) ]~,
\end{equation}
also forms a valid quantum channel.
\end{defi}
The above definition naturally generalizes to tristochastic channels. {This and further definitions in this subsection are borrowed \rd{from  or inspired by}~\cite{bistron2023tristochastic}.}

\begin{defi}\label{tri_q_def}
\rd{A quantum channel $\Phi_D: \Omega_{d^2} \to \Omega_d$} defined by dynamical matrix $D_{j_1, j_2, j_3}^{i_1, i_2, i_3}$ by

\begin{equation}
\Phi_D[\rho_2  \otimes \rho_3] = \Tr_{2, 3}[D (\mathbb{I} \otimes \rho_2^\top \otimes \rho_3^\top)]~,
\end{equation}
is called a tristochastic channel, if for any pair of density matrices $\{\rho_1, \rho_2\}$ the maps:
\begin{equation}\label{eq_tri_D2}
\Tr_{1,3}\;[D(\rho_1^\top \otimes \mathbb{I} \otimes \rho_{2}^\top) ] \text{~and~} \Tr_{1,2}\;[D(\rho_1^\top \otimes \rho_2^\top \otimes \mathbb{I} ) ],
\end{equation}
also forms a valid quantum channel.
\end{defi}

An alternative, and much wider, approach to \rd{translate} tristochasticity at the quantum level comes from the notion of coherification~\cite{Korzekwa_coch} which aims to promote classical-probabilistic objects onto a quantum level. It is achieved by following the idea that the diagonals of density matrices are treated as a classical probabilistic vector.

\begin{defi}
\label{def:q_tristoch}
A coherification of a tristochastic tensor $A$, is a channel \rd{$\Phi_A:\Omega_{d^2} \to \Omega_d$}, such that the diagonal of its dynamical matrix $D$ agrees with the elements of $A$,
\begin{equation}
\forall_{k, l, j }~~ D_{k,l,j}^{k,l,j} = A_{k l j},
\end{equation}
with the CPTP properties of the channel $\Phi_A$ guaranteed by the positivity $D \geq 0$ and trace condition $\Tr_1[D] = \id_{d^2}$.
\end{defi}
In other words, the coherification procedure is a search for preimages of the physical process of decoherence for the dynamical matrix $D$.

The recipe for coherification is ambiguous. Thus, following~\cite{bistron2023tristochastic} from the entire set of coherifications of certain permutation tensor, we chose only those with maximal $2$-norm coherence $C_2$~\cite{Korzekwa_coch}, defined by the sum of nondiagonal elements of $D$ modulus squared. It turns out that this choice is very fruitful, leading to the main object of the presented work:

\begin{defi}\label{def:convololo_channel}
Let $A$ be a tristochastic permutation tensor. The convolutional channel $\Phi_A$ associated with $A$ is a coherification of $A$ with maximal $2$-norm coherence of its dynamical matrix $D$.
Moreover \rd{we can identify convolutional channel $\Phi_A$ with a bipartite unitary matrix \cite{bistron2023tristochastic}}:
\begin{equation}
\label{U_form1}
U_{ki,lj} = A_{klj} |a_{k,l}\ra_{i}
\end{equation}
where $A_{klj}$ is the initial permutation tensor and $|a_{k,l}\ra$ are complex vectors satisfying $\la  a_{k,l}|a_{k,l'}\ra = \delta_{l,l'}$, followed by a partial trace:
\begin{equation}
\label{eq:u_coch_21}
    \Phi_A[\rho_1 \otimes \rho_2] = \Tr_2[U(\rho_1 \otimes \rho_2)U^\dagger]~.
\end{equation}
\end{defi}

The dynamical matrix for the convolutional channel associated with $A$ has a simple form
\begin{equation}
\label{D_tristoch1}
D_{k,\;l,\;j}^{k',l',j'} = A_{klj} A_{k'l'j'} \la a_{k,l}|a_{k',l'}\ra = \sum_{i} U_{ki,lj} \bar{U}_{k'i,l'j'},
\end{equation}
with unitary matrix $U$ as in equation \eqref{U_form1} and $\bar{U}$ denoting complex conjugate of $U$. \rd{We can explicitly verify that 
\begin{equation*}
D_{k,l,j}^{k,l,j} = A_{klj} A_{klj} \la a_{k,l}|a_{k,l}\ra = A_{klj}.
\end{equation*}
On the other hand, any unitary defining a coherification of $A$ via \eqref{eq:u_coch_21} must satisfy
\begin{equation*}
A_{klj} = D_{k,l,j}^{k,l,j} =\sum_{i} U_{ki,lj} \bar{U}_{ki,lj} = \sum_i |U_{ki,lj}|^2.
\end{equation*}
Thus entries of $U_{ki,lj}$ can be nonzero only if $A_{klj}$ is nonzero which, together with unitarity requirement leads to the form \eqref{U_form1}. This shows that the above definition is self-consistent. For further discussion of optimal coherifications of (multi)stochastic tensors we refer the reader to Appendix \ref{app:max_coch}.}

Therefore, from this point on, we will identify the convolutional channel $\Phi_A$ 
with corresponding unitary channel $U$. Notice that in the proposed channel each basis $\{|a_{k,l}\ra\}_{l = 1}^{d}$ can be freely adjusted to a specific task,  without losing general properties of the convolutional channel. Thus one has a lot of free parameters handily combined in $d \times d$ unitaries corresponding to the basis $\{|a_{k,l}\ra\}_{l = 1}^{d}$.

\subsection{Entangling power of unitary operations}

To characterize the properties convolutional channels $\Phi_A$ 
discussed above, or equivalently $d^2$ unitary matrices \eqref{U_form1}, from the perspective of their ability to entangle and disentangle quantum systems, we recall the framework of entangling power of bipartite unitary gates and related notions.

\begin{defi}\cite{Zanardi2000entangling}\label{ep_def}
Consider a unitary operation $U$ acting on a bipartite space $\mathcal{H}_{AB} \equiv \mathcal{H}_{A}\otimes \mathcal{H}_B$ \rd{both} with local dimension $d$. The entangling power $e_p$ is defined as the average entanglement created by $U$ when acting on a pure product state $\ket{\psi_A}\otimes\ket{\psi_B}\in\mathcal{H}_{AB}$,
\begin{equation}
    e_p(U) =\! \frac{d+1}{d-1} \ev{\mathcal{E}\qty[ U(\ket{\psi_A}\otimes\ket{\psi_B})]}_{\ket{\psi_A},\ket{\psi_B}}
\end{equation}
where the average $\ev{\cdot}_{\ket{\psi_A},\ket{\psi_B}}$ is taken over Haar measure with respect to both subspaces and $\mathcal{E}$ is an entanglement measure. 
\end{defi}
In what follows we will use the linear entropy as the measure of entanglement, $\mathcal{E}\qty(\ket{\psi}) = 1 - \Tr(\rho_A^2)$, with \mbox{$\rho_A \equiv \tr_B\op{\psi}$}. The normalization in Definition \ref{ep_def} follows from the requirement $e_p \in [0,1].$ 
~A closed formula for entangling power was obtained in~\cite{zanardi2001entanglement}
\begin{equation}
e_p(U) = \frac{\mathcal{E}\qty(\ket{U}) + \mathcal{E}\qty(\ket{US}) - \mathcal{E}\qty(\ket{S})}{\mathcal{E}(\ket{S})}
\end{equation}
where the state $\ket{U} = \qty(U\otimes I)\ket{\Psi_+}$ is defined according to the Choi-Jamiołkowski isomorphism \rd{with $\ket{\Psi_+} = \frac{1}{\sqrt{d}} \sum_{i=1}^d \ket{ii}$ being the Bell state defined on $\mathcal{H}_{AB}^{\otimes 2}$} and $S = \sum_{i,j=1}^d \op{ij}{ji}$ is the SWAP operator. \rd{In this and following formulas whenever we consider entanglement of ``vectorized unitary'' $|U\ra$, the bipartition $\mathcal{H}_A^{\otimes2} \otimes \mathcal{H}_B^{\otimes 2}$ is assumed. Thus entanglement of identity $\mathcal{E}(|I_{AB}\rangle) = \mathcal{E}(|I_A\rangle\otimes |I_B\rangle)$ is equal $0$ and the entanglement of swap $\mathcal{E}(|S\ra) = (d^2-1)/d^2$ is maximal.}

A unitary operation $U$ is called \textit{2-unitary}, if the entangling power is maximal, $e_p = 1$ or, equivalently, if partial transpose $(U^{\Gamma})_{ki,lj} = U_{li,kj}$ and reshuffling $(U^R)_{ki,lj} = U_{kl,ij}$ are also unitary~\cite{permutatrion_entangling_power}. 

Further two quantities of interest are \textit{gate typicality} $g_t$ (complementary to $e_p$) and \textit{disentangling power} $d_p$\footnote{Note that this notion of disentanglement should not be confused with a substantially different concept present within the CNN community~\cite{chen2018isolating}.} (defined as average entanglement left after the action of the gate on an arbitrary maximally entangled state) of a bipartite unitary matrix. They can both be given by the following formulae~\cite{CLARISSE2007400, jonnadula2017impact}
\begin{equation}
    \begin{aligned}
        g_t(U) & = \frac{\rd{\mathcal{E}\qty(\ket{U}) - \mathcal{E}\qty(\ket{US}) + \mathcal{E}\qty(\ket{S})}}{2\mathcal{E}(\ket{S})}, \\
        d_p(U) & = \frac{1}{d-1}e_p(U). 
    \end{aligned}
\end{equation}
\rd{Disentangling power characterizes how much local subsystem are ``exchanged'' by unitary $U$. It is best understood by taking a product state $\rho_A\otimes\rho_B$ and considering two extreme cases. First, acted upon by product unitary, with $g_t = 0$, the output of $A$ party is full determined by $\rho_A$, and likewise for $B$. In the second case, we use SWAP gate with $g_t = 1$, and it is immediate to note that output $A$ is determined in terms of state $\rho_B$ and \textit{vice versa}, showing complete interchange of influence. Since for large local dimensions the distribution of $g_t$ is strongly concentrated at the mean value $\la g_t\ra = 1/2$, substantial deviation from the mean, which also implies small entangling power \cite{jonnadula2017impact}, indicate that the action of bipartite unitary is similar (up to local unitary operations) to identity or swap.} 

Note also that entangling power and disentangling power are proportional~\cite{CLARISSE2007400}, so the channels with maximal entangling power have also maximal disentangling power.

All the aforementioned properties of bipartite unitaries are invariant under local operations. For example, if for some unitary $U$ one has maximal entangling power, $e_p(U) = 1$, then for any local rotations $v_1, v_2, v_1', v_2'$, $e_p[(v_1 \otimes v_2)U(v_1' \otimes v_2')] = 1$. The work~\cite{Suahil_invariants} presents a helpful tool to verify whether two unitaries $U$ and $U'$ can be connected by \rd{local operations in such a way}.

\begin{defi}\cite{Suahil_invariants}
Let $U_{kl}^{ij}$ be a 2-unitary, and $\sigma, \tau,\rho,\lambda$ be $n$-element permutations. Then the invariant of the local rotations $I_{\sigma, \tau,\rho,\lambda}(U)$ is given by:
\begin{equation} \label{eq:sukhail_invariant}
I_{\sigma, \tau,\rho,\lambda}(U) = U_{k_1l_1}^{i_1j_1}\cdots U_{k_nl_n}^{i_nj_n} \;\; \overline{U}_{k_{\rho(1)} l_{\lambda(1)}}^{i_{\sigma(1)} j_{\rho(1)}} \cdots \overline{U}_{k_{\rho(n)} l_{\lambda(n)}}^{i_{\sigma(n)} j_{\rho(n)}}
\end{equation}
where the sum over repeated indices is assumed.
\end{defi}

Thus if two unitaries $U$, and $U'$ have different invariants, then one cannot be transformed into the other by local pre- and post-processing. Moreover, if all invariants $I_{\sigma, \tau,\rho,\lambda}(\cdot)$ for $U$ and $U'$ have the same values, then they can be connected by local operations~\cite{Suahil_invariants}.

\subsection{Orthogonal Latin squares}

The last construction we refer to is a notion connected to combinatorics --  \textit{Latin squares}~\cite{LatinHypercubes}.

\begin{defi}
A Latin square $L$ of dimension $d$ is a $d \times d$ matrix with entries from the set $[d] := \{1, \cdots,d\}$ such that in each row and each column contains all the elements of the set $[d]$.
\end{defi}

In other words, each column and each row contain all the numbers from $1$ to $d$ without repetitions, as can be seen in an exemplary Latin square of size 3 below,

\begin{equation}
\label{example_Latin_square}
    L = \mqty(
        1 & 2 & 3 \\
        3 & 1 & 2 \\
        2 & 3 & 1).
\end{equation}

To obtain a permutation tensor $A$ from the Latin square $L_{jk}$ one may simply set $A_{ijk} = \delta_{i,L_{jk}}$.
Examples of corresponding permutation (tri)stochastic tensor and Latin square are \eqref{example_permutation_tensor} and \eqref{example_Latin_square}.

Together with the concept of Latin squares comes also the notion of their \textit{orthogonality}.

\begin{defi}
Two Latin squares, $L$ and $M$, are said to be orthogonal if the set of pairs $\{(L_{ij} , M_{ij} )\}$ has $d^2$ distinct elements.
\end{defi}

It is impossible to find two orthogonal Latin squares of dimension 2. However, for the exemplary Latin square~$L$~\eqref{example_Latin_square} we can give an orthogonal Latin square

\begin{equation}
    M = \mqty(
        1 & 2 & 3 \\
        2 & 3 & 1 \\
        3 & 1 & 2).
\end{equation}

For any prime and prime-power dimension $d$ there exists a construction based on finite fields, yielding exactly $d-1$ pairwise-orthogonal Latin squares~\cite{denes1974Latin}. It is known~ that for 
any $d\ge 7$ there exist at least two orthogonal Latin squares.

It has been shown in~\cite{CLARISSE2007400} that given two orthogonal Latin squares $L$ and $M$ in dimension $d$ one can construct a 2-unitary permutation matrix $P_{d^2}$ with  $e_p = 1$ in a straightforward way,
\begin{equation}
\label{U_orto_Latin_s}
    P_{d^2} = \sum_{l,j=1}^d \op{L_{lj},M_{lj}}{l,j}.
\end{equation}

Such a construction is a nice example convolutional channel associated with permutation tensor \mbox{$A_{klj} = \delta_{k, L_{lj}}$}, with basis vectors $|a_{k,l}\ra$ given by:
\begin{equation}
\label{ort_Latin_basis}
    \left(a_{k,l}\right)_i = \delta_{i,M_{lj}} \text{ with } j \text{ such that } k = L_{lj}.
\end{equation}

\section{Coherence range for bipartite unitaries}\label{sec:coch_measure}

To quantitatively distinguish the presented construction from previously known ones we adapt the notion of coherence to describe unitary channels. 
In order to quantify a property akin to coherence for unitary operations we will consider average coherence generated by the action of a unitary matrix on the basis vectors, \rd{quantified} by $\alpha$-R{\'e}nyi entropies applied to the amplitudes of the resulting states\rd{, which have been considered as coherence monotones for states eg. in \cite{ChitambarGour2016, Rastegin2016, zhu2017coherence, Audenaert2015}.}

Let us take an arbitrary pure state $\ket{\psi}\in
\mathcal{H}_{d\times d}$. Its coherence with respect to the basis defined by a unitary matrix $U$ may be characterized as
\begin{equation}
\label{eq_whatever1}
H_\alpha(\ket{\psi};U) = \frac{1}{1-\alpha}\log(\sum_{i=1}^D \abs{\mel{i}{U}{\psi}}^{2\alpha}),
\end{equation}
where $D = d^2$. \rd{It can be seen as measured R{\'e}nyi $\alpha$-entropy. Additionally, in restriction to pure states it has been shown equivalent to quantum R{\'e}nyi $\beta$-entropy relative to incoherent states with $\beta = \frac{\alpha}{2\alpha - 1}$ \cite{ChitambarGour2016}.} 
For $\alpha \in \qty{0, 2, \infty}$  exponentials of these entropies,
\begin{equation}
\label{eq_whatever2}
\begin{aligned}
    S_\alpha(\ket{\psi};U) & = \exp[(1-\alpha)H_\alpha(\ket{\psi};U)] \\
    & = \sum_{i=1}^D \abs{\mel{i}{U}{\psi}}^{2\alpha}
\end{aligned}
\end{equation}
turn out to have simple interpretations related to the number of nonzero elements, linear entropy and maximal element -- for details see Appendix \ref{app:coh_measure}. \rd{It is important to note, however, that R{\'e}nyi entropy in the limit $\alpha \to 0$ is not a coherence monotone in a strict sense and can be treated at best as a coherence indicator, appropriate for pure states, as they are incoherent if and only if they have a single non-zero element.} 

Using the above notion, we may define a measure of coherence of a unitary matrix $U$ as an average coherence of computational basis vectors:
\begin{equation}
\label{eq_whatever3}
    S_\alpha(U) = \frac{1}{D}\sum_{j = 1}^D S_\alpha(\ket{j};U).
\end{equation}

Bipartite 2-unitary matrices offer an important degree of freedom to make use of, as two such matrices $U$ and $U'$ are considered {\sl locally equivalent} if
there exist local unitaries $v_1$, $v_2$, $v_1'$, $v_2'$ such that \mbox{$(v_1\otimes v_2) U (v_1'\otimes v_2') = U'$}. Therefore each two-unitary corresponds to the entire range of $S_{\alpha}(U)$ values:

\begin{widetext}
\begin{equation}
\label{range_def1}
    {\operatorname{range}}\big(S_\alpha(U)\big) = \qty{\min_{V, V' \in U(d)^{\otimes 2}} S_\alpha(VUV'),\max_{V, V' \in U(d)^{\otimes 2}} S_\alpha(VUV')}.
\end{equation}
\end{widetext}
\bl{with extrema taken over tensor products of local unitary operators from $U(d)$ i.e. $V = v_1\otimes v_2$, $V' = v_1' \otimes v_2'$.}

Let us consider two simple examples, starting with a case of local unitary operation, $U = U_A\otimes U_B$, evaluated with respect to the linear entropy. One easily finds that

\begin{equation}
    \operatorname{range}\big(S_2(U_A\otimes U_B)\big) = \qty{\frac{1}{d^2},1}.
\end{equation}
The maximum is found by setting $V = U^{\dagger}$ and keeping $V' = \mathbb{I}$, while minimum is found when \mbox{$V = F_d^{\otimes2} U^{\dagger}$}, where $F_d$ is the Fourier matrix of dimension $d$, \mbox{$\qty(F_d)_{jk} = \frac{1}{\sqrt{d}}\exp(\frac{2\pi i}{d}jk)$}. 

The same holds also if $U = P$ is an arbitrary permutation matrix and, in particular, a 2-unitary matrix constructed from two orthogonal Latin squares
\begin{equation}
    \operatorname{range}(S_2(P)) = \qty{\frac{1}{d^2},1}~,
\end{equation}
which are values for a ``bare'' permutation ($V = V' = \mathbb{I}$) and Fourier matrices ($V = F_d^{ \otimes2}$, $V' = \mathbb{I}$), respectively. For a generic $U$ such range is not \textit{a priori} trivial, in particular the maximal value points towards the non-vanishing coherence of the matrix, understood as ability of a given gate to generate nonzero coherences for any choice of local bases.

In some scenarios one may desire that in their circuit the average coherence of the output would lie in some particular range, independent of local pre- and post-processing of input basis states. Such a case corresponds exactly to a narrow coherence range.

For a more detailed discussion of this construction, further motivations and its properties, we encourage the readers to consult Appendix \ref{app:coh_measure}.

\section{Properties of Convolutional
Channels}\label{sec:tristoch}

In this section, we study the general features of an entire class of convolutional channels.
As our main quantities of interest, we choose entangling power $e_p$ and gate typicality $g_t$, because they allow us to \rd{characterize properties important from} the perspective of handling entangled states.

Entangling power $e_p$ of a \rd{unitary} $U$ of size $d^2 \times d^2$ describes, how much the outcome $U (|\psi\ra \otimes |\phi\ra)$ is entangled on average for random input pure states $|\psi\ra$ and $|\phi\ra$. Hence, after the partial trace, it gives us insight into how much the result of the channel $\Phi_A$ becomes mixed.
On the other hand, the gate typicality $g_t$ describes the degree of subsystem exchange which, taking into account the partial trace, reveals which subsystem affects the output more. 

Let us now consider $U$ corresponding to a convolutional channel $\Phi_A$ defined by a basis $\qty{\ket{a_{k,l}}}$. By simple calculations, one obtains:
\begin{equation}\label{max_ep_cond}
\begin{aligned}
\mathcal{E}\qty(\ket{U}) & =1 - \frac{1}{d^4} \!\!\sum_{k,l,k',l',j}\!\!\!\! A_{klj} A_{k'l'j} |\la a_{k,l}| a_{k',l'}\ra|^2 , \\
\mathcal{E}\qty(\ket{US}) & = 1 - \frac{1}{d^4} \sum_{k,k',l} |\la a_{k,l}| a_{k',l}\ra|^2. \\
\end{aligned}
\end{equation}
This allows us \rd{to} establish bounds for entangling power and gate typicality for the aforementioned unitary matrices:
\begin{equation}
\begin{aligned}
\label{ep_bound}
1 - \frac{1}{d+1} & \leq e_{p}(U) \leq 1, \\
\frac{1}{2} - \frac{1}{2d + 2} & \leq g_t(U) \leq \frac{1}{2} + \frac{1}{2d + 2}.
\end{aligned}
\end{equation}
Moreover, we obtained also the average values  of entangling power and gate typicality over all convolutional channels:
\begin{equation*}
\la e_p(U) \ra_{|a\ra} = 1 - \frac{2}{d^2 + d} ~,~~~~~ \la g_t(U)\ra_{|a\ra} = \frac{1}{2} ~.
\end{equation*}
Note that the lower bound for entangling power $e_p$ for unitaries corresponding to convolutional channels coincides with the upper bound of entangling power for block diagonal bipartite unitary matrices with $d$ blocks of size $d\times d$~\cite{dual_uni_to_bernoulli_circuits}.
The exact derivation of those results is provided in Appendix  \ref{app:q_tri_call}. The minimal values of entangling power $e_p$ are obtained for example when $|a_{k,l}\ra = |a_{k',l}\ra$ for any $k,k'$ which corresponds to minimal gate typicality, or for $|a_{k,l}\ra_i = A_{kli}$, which corresponds to maximal gate typicality, see Fig.~\ref{fig:ep}.

One can also consider the case when all the bases $\{|a_{k,l}\ra\}_{l = 1}^d$ are mutually unbiased (MU)~\cite{mub_paper} i.e. for any two vectors from different basis $|a_{k,l}\ra$ and $|a_{k',l}\ra$ the overlap between basis vectors is given by
\begin{equation}
\label{MU_def}
    \abs{\braket{a_{kl}}{a_{k'l'}}}^2 = \frac{1}{d}.
\end{equation}
Taking $\{|a_{k,l}\ra \}$ to be a set of MU bases (MUBs),
and defining, in turn, $U_{\text{MUB}}$ as the unitary representation of convolutional channel, resulting from Eq.~\eqref{U_form1} with such choice, the entangling power and gate typicality are given by: 
\begin{equation*}
 e_p(U_{\text{MUB}})  = 1 - \frac{2}{d^2 + d} ~,~~~~~  g_t(U_{\text{MUB}}) = \frac{1}{2} ~,
\end{equation*}
which are exactly the average values of $e_p$ and $g_t$.

\begin{figure}[h]
    \includegraphics[height=4.2in]{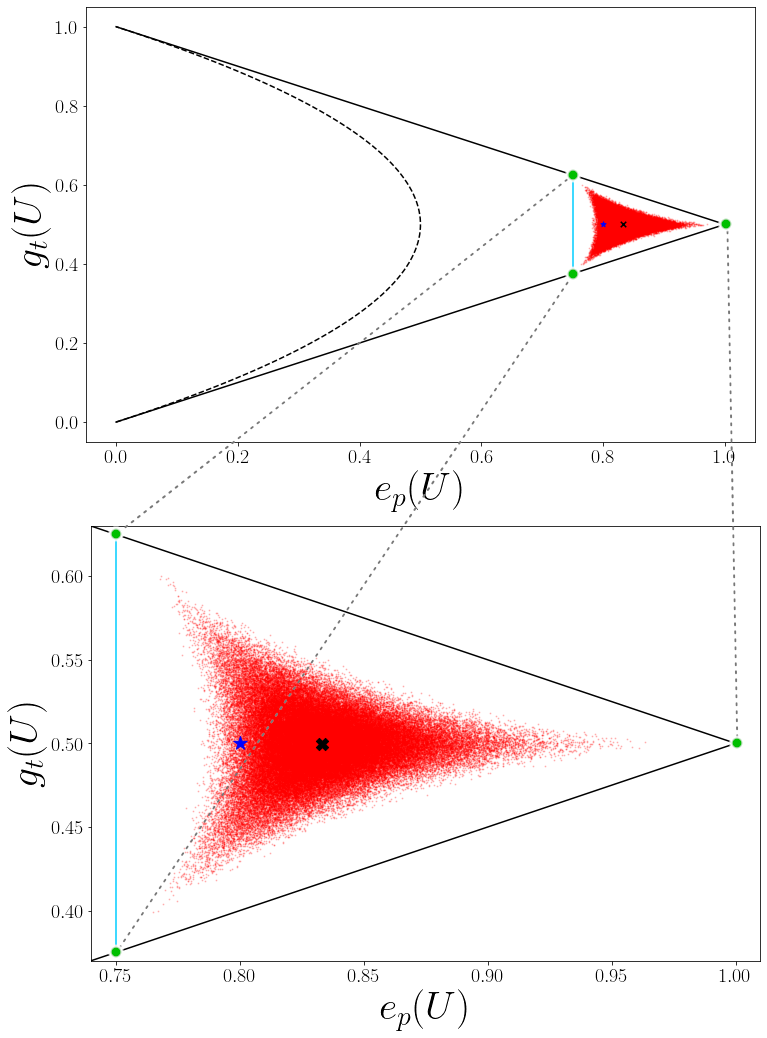}
    \caption{Plots of entangling power and gate typicality for dimension $d = 3$. Black lines correspond to the general bounds for unitary channels (dashed line 
    generated by powers of the swap operation $S^\kappa$), and \rd{cyan} lines mark the lower bound of entangling power for convolutional channels \eqref{ep_bound}.  The cloud of red dots corresponds to a random choice of $10^5$ bases $\{|a_{(kl)}\ra\}_{l = 1}^d$, with each basis taken from Haar measure on $U(d)$, green points to the extremal cases, black cross corresponds to the MUB case and blue star to entangling power and gate typicality averaged over all unitary matrices of from CUE of size $d^2$, see \eqref{standard_average_ep}. 
    The lower plot shows a magnification of the entire region presented in the upper panel. }
    \label{fig:ep}
\end{figure}

The case of maximal entangling power is the most complex one. Before discussing it in detail let us reestablish the connection with (fully quantum) tristochasticity.

\begin{theo}\label{theo_ep_stoch}
A unitary matrix $U$ corresponding to a convolutional channel $\Phi_A$ has the maximal entangling power $e_p(U) = 1$ if and only if the channel $\Phi_A$ is tristochastic. 
\end{theo}

We present a proof of this theorem in Appendix \ref{app:q_tri_call} for the clarity of the text. \rd{Furthermore, in Appendix \ref{app:max_coch} we generalize the above theorem for arbitrary tristochastic tensors, for which the coherification is no longer explicit.}
To capture the connection between entangling power $e_p$ and (quantum) tristochasticity let us alter the notation $|\vb{a}_{klj}\ra:= A_{klj}|a_{kl}\ra$ (with no summation involved)  to create the symmetry between indices $k,l,j$. Then the condition for unitarity of $U$ gives 
\begin{equation}
\label{multi_tri1}
\forall_k\forall_{ll'j'j}~~ \la \vb{a}_{klj}|\vb{a}_{kl'j'}\ra = \delta_{ll'}\delta_{jj'}~,
\end{equation}
and the maximal entangling power imposes, from \eqref{max_ep_cond1},

\begin{align}
\label{multi_tri2}
&\forall_j \forall_{kk'll'}~~ \la \vb{a}_{klj}|\vb{a}_{k'l'j}\ra = \delta_{kk'}\delta_{ll'}~, \\
\label{multi_tri3}
&\forall_l \forall_{kk'jj'}~~ \la \vb{a}_{klj}|\vb{a}_{k'lj'}\ra = \delta_{kk'}\delta_{jj'}~.
\end{align}

On the other hand, one may interpret \eqref{multi_tri2} or \eqref{multi_tri3} as the conditions for unitarity of the following matrices

\begin{equation}
\label{new_tri_u}
U_{ji,kl}' = |\vb{a}_{klj}\ra_i \text{ and } U_{li,jk}'' = |\vb{a}_{klj}\ra_i~,
\end{equation}
\rd{each with maximal entangling power as well}. These, composed with the partial trace, gives the \rd{complementary} tristochastic channels as in the Definition \ref{tri_q_def}.

More intuitively, one can imagine a $d\times d \times d$ cube populated by $d^2$ states $\ket{a_{kl}}$ placed in positions corresponding to nonzero $A_{klj}$ elements in such a way that each horizontal, vertical and ``depth'' slice contains an orthonormal basis on $\mathcal{H}_d$. This can be satisfied because there is only one non-zero state in each column, row and ``depth-row'', since $A_{klj}$ is a permutation tensor.

The allowed region of entangling power $e_p$ and gate typicality $g_t$ for the convolutional channel unitary $U$ \eqref{U_form1} together with extremal, distinct and randomly sampled unitaries are presented in Fig.~\ref{fig:ep}. 

\begin{figure}[h]
        \includegraphics[height=2.1in]{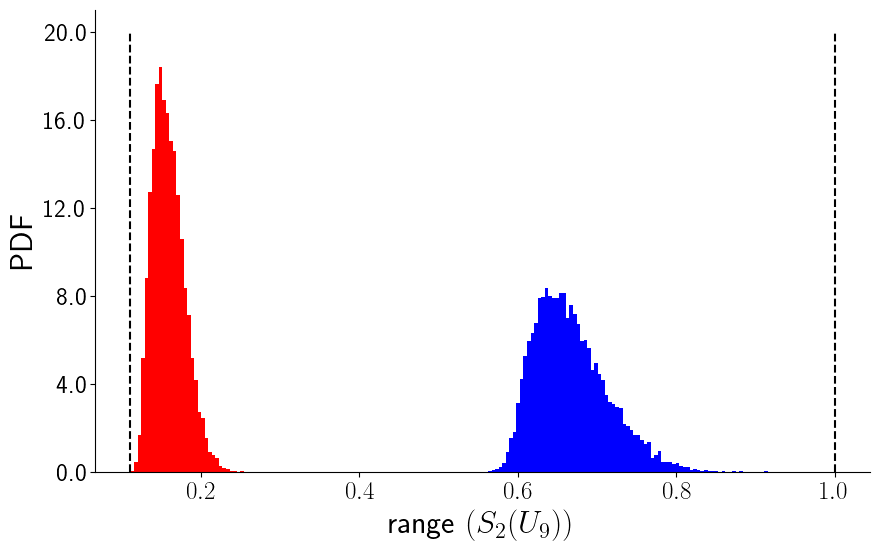}
    
        \includegraphics[height=2.1in]{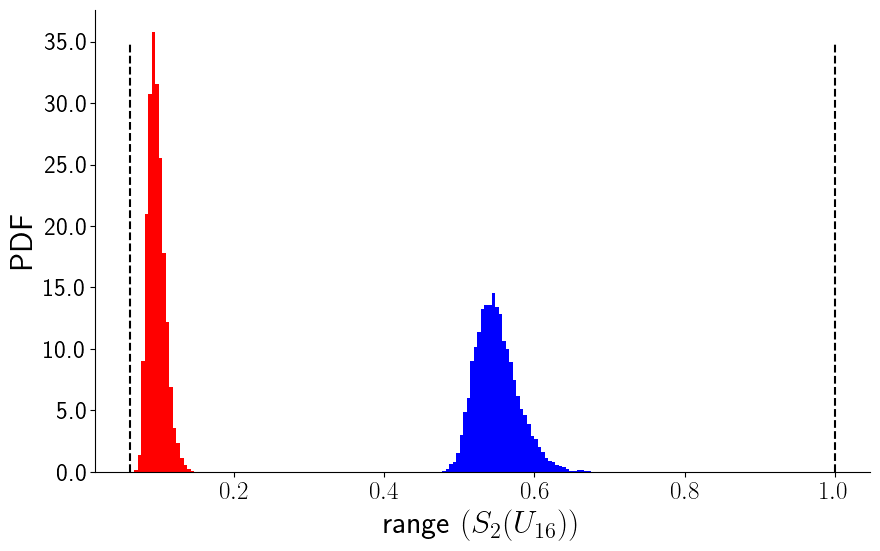}
    \caption{
    Probability density functions (PDFs)
    of estimated ranges of coherence, $\text{range}(S_{2}(U))$, for convolutional channels in dimensions $3\times 3$ on top and $4 \times 4$ in the bottom. In each dimension, we generated coherifications by drawing $10^4$ random basis $|a_{(kl)}\ra$ and accumulated estimated lower bound $S_{2}^{min}(U)$ (red) and upper bound $S_{2}^{max}(U)$ (blue) of $\text{range}(S_{2}(U))$. Dotted lines correspond to the general bounds of $S_2$ covered by permutations. }
\end{figure}

\section{New classes of genuinely quantum 2 unitary gates}\label{sec:new}
In this section, we introduce new continuous families of 2-unitary gates emerging from convolutional channels. 

The dimension $d = 2$ was already discussed in ~\cite{bistron2023tristochastic}. For the peculiar dimension $d = 6$, for which first 2-unitary matrix was found only recently~\cite{36_officers_of_Karol}, despite extensive numerical searches, we did not find a coherification of any permutation tensor with entangling power larger than $e_p = \frac{208 + \sqrt{3}}{210}\approx0.9989$ achieved in~\cite{bruzda2022structured}. In Appendix \ref{app:orto6} we present an exemplary \textit{orthogonal} matrix corresponding to convolutional channel, that achieves this bound and could serve as a candidate for the most entangling orthogonal gate of order $36$. 
One might think, that the construction \eqref{U_orto_Latin_s} encompasses all the convolutional channels with maximal (dis)entangling power since it is the case for $d = 3,4,5$, which can be verified by direct (exhaustive) calculations. Nevertheless, in dimension $d\geq 7$ it is possible to find nontrivial sets of bases $\ket{a_{k,l}}$ which generate unitary channels with maximal entangling power.

In order to present our construction let us first rewrite (\ref{multi_tri1}--\ref{multi_tri3}) as a conditions of unitarity for $3d$ matrices, defined by
\begin{equation}
\begin{aligned}
\label{uni_ame_cond}
&(V_k)_{li} :=\sum_j\ket{\vb{a}_{klj}}_i =  \ket{a_{kl}}_i ~, \\
&(V'_k)_{ji}  := \sum_l\ket{\vb{a}_{klj}}_i~, \hspace{1 cm} (V''_l)_{ji} = \sum_k \ket{\vb{a}_{klj}}_i~,  \\
\end{aligned}
\end{equation}
where in each sum there is in fact only one nonzero vector due to $A_{klj}$ being a permutation tensor.

Thanks to the structure above we may leverage an algorithm akin to the Sinkhorn approach used in~\cite{36_officers_of_Karol}. Cycling through the sets $\qty{V_k},\,\qty{V'_k},\,\qty{V''_k}$, we orthonormalize each element of the set using polar decomposition. The standard complexity of polar decomposition of a matrix of size $D$, using SVD, behaves as $O\qty(D^3)$. When full matrix $U$ is considered as in~\cite{36_officers_of_Karol} ($D = d^2$),  this results in complexity $O\qty(d^6)$. The proposed approach relies on the decomposition of $d$ square matrices $V_k$, of size $d$ each, resulting in complexity  $O\qty(d^4)$. 

Furthermore, our approach limits the dimensionality of the space explored and as such, should be faster to converge than the methods employed in the full space\footnote{Assuming that a solution exists within the limited space, which is not guaranteed.}.
\rd{It is so, because the} number of complex parameters, thanks to the choice of the tristochastic tensor $A_{klj}$ and local freedom to fix $V_1 = \id_d$, is equal to $d(d-1)(d-2)$, whereas the previous approaches have, in general, $d^4$ parameters, again highlighting reduction of complexity.

In particular, by focusing our attention on cyclic permutation tensors $A_{klj} = \delta_{k,l \oplus j}$ 
and their coherifications constructed using bases with cyclic amplitude structures,
\begin{equation}
    \abs{\braket{i}{a_{k,l}}}^2 =
    \abs{\braket{i\oplus n}{a_{k,l \oplus n}}}^2, 
\end{equation}
where $\oplus$ denotes the addition modulo $d$,
together with fixing the first basis as the computational basis, \mbox{$\ket{a_{1,l}} = \ket{l}$},
we were able to derive two novel continuous families of 2-unitary matrices of dimension $d^2$. In the following subsections we will discuss a 2-parameter family for dimension $d=7$ and a family for dimension $d=9$ characterized by two $3\times3$ bases with cyclic structure. 
The aforementioned new classes of $2$-unitary matrices, which will be presented in the following subsections, are genuinely quantum~\cite{36_officers_of_Karol}, in the sense, that they are not locally equivalent to any permutation tensor $P_{d^2}$ 
from equation \eqref{U_orto_Latin_s}:
\begin{equation*}
U_{d^2}\neq (v_1\otimes v_2)P_{d^2}(v_1'\otimes v_2')    
\end{equation*}
for any local pre- and postprocessing.

Both of the presented classes in dimensions $7$ and $9$, similarly as permutations~\cite{Suahil_invariants}, can be further extended by the multiplication 
by a diagonal unitary matrix with arbitrary phases, giving additional $36$ and $65$ nonlocal parameters, respectively, (corresponding to a number of phases that cannot be removed by local transformation), which we call the "simple" phasing of matrix elements.

\subsection{Dimension 7}

The two parameter family of 2-unitary operations $U_{49}(\phi_1,\phi_2)$ of order $d^2$ found for $d = 7$ is characterized by seven bases $\ket{a_{k,l}}$, which are presented in Fig.~\ref{fig:dimension_7}. After fixing the first basis as equal to the computational basis using local rotation, the remaining six bases are particularly elegant and can be characterized by just two non-zero amplitudes: $\sqrt{1/7}$ and $\sqrt{2/7}$, nine distinct constant phases from the set $\qty{\pm\arccos(\pm\frac{3}{4}),\,\pm\arccos(\frac{\pm1}{2\sqrt{2}}),\pi}$, and two free phases $\phi_1,\phi_2\in[0,2\pi)$, as summarised in Fig. \ref{fig:dimension_7}. 

In order to compare this family to the already known solutions based on Latin squares we have resorted to the 4-th order invariant $I_{\sigma, \tau, \rho, \lambda}$ defined in \eqref{eq:sukhail_invariant} with the permutations
\begin{align*}
   \sigma & = \text{Id},  & \tau = (12)(34), \\
   \rho & = (13)(24), & \lambda = (14)(23),
\end{align*}
in line with the invariant used in~\cite{Suahil_invariants} for the case of 36 officers. The analytical expression 
is given by

\begin{widetext}
\vspace{-0.05 cm}
{\scriptsize
    \begin{equation}
        \begin{aligned}
        I_{\sigma, \tau, \rho, \lambda}(U_{49}) = & \frac{1}{7} (9614-4 \sqrt{7} \sin(\phi _1)-4 \sqrt{7} \sin(\phi _1-\phi _2)-3 \sqrt{7} \sin(2 \phi _2)-2 \sqrt{7} \sin(\phi _1+\phi _2)-6 \sqrt{14} \sin(\phi _1+\phi _2)-2 \sqrt{7} \sin(2 \qty(\phi _1+\phi _2)) \\
        & -4 \sqrt{14} \sin(\phi _1+2 \phi _2)+6 \sqrt{7} \sin(2 \phi _1)+14 \sqrt{7} \sin(\phi _2)+2 \sqrt{7} \sin(2 \qty(\phi_1+2 \phi _2))+2 \sqrt{14} \sin(\phi_1) +2 \sqrt{14} \sin(\phi _2) \\
        & +\sqrt{14} \sin(2 \phi _2)+2 \sqrt{14} \sin(2 \qty(\phi _1+\phi _2))+12 \sqrt{14} \sin(2 \phi _1+\phi _2)+18 \cos(2 \phi _1)-2 \qty(5 \sqrt{2}+2) \cos(\phi _1)-68 \cos(\phi _1-\phi _2) \\
        & -58 \cos(\phi _2)-3 \sqrt{2} \cos(2 \phi _2)+17 \cos(2 \phi _2)-10 \sqrt{2} \cos(\phi _1+\phi _2)+2 \cos(\phi _1+\phi _2) -10 \sqrt{2} \cos(2 \qty(\phi _1+\phi _2)) \\
        & -18 \cos(2 \qty(\phi _1+\phi _2)) -4 \sqrt{2} \cos(2 \phi _1+\phi _2)-8 \sqrt{2} \cos(\phi _1+2 \phi _2)-6 \cos(2 \qty(\phi _1+2 \phi _2))+10 \sqrt{2} \cos(\phi _2)).
        \end{aligned}
    \end{equation}}
    \vspace{-0.3 cm}
\end{widetext}

Global minima and maxima of this function can be easily found, thus bounding it by
\begin{equation}
    1347.84 \leq I_{\sigma, \tau, \rho, \lambda}(U_{49}) \leq 1403.66
\end{equation}
with the lower bound being well above the same invariant calculated for any 2-unitary permutation $P_{49}$ in dimension $d = 7$, equal to $I_{\pi,\tau,\rho,\lambda}(P_{49}) = 7^3 = 343$. This is enough to demonstrate that the entire family $U_{49}(\phi_1,\phi_2)$ is locally inequivalent to any 2-unitary permutation $P_{49}$.

In order to bound the coherence $\text{range}(S_2(U_{49}))$ from the inside we calculated $S_2(U_{49})$ as approximation of upper limit and $S_2(F^{\otimes 2}\,U_{49})$ with the Fourier $F_{jk} =\frac{1}{\sqrt{7}} e^{\frac{2i\pi}{7}jk}$ as the lower limit of the approximation. Interestingly, for all $\phi_1, \phi_2$ we find $S_2(U_{49}) = \frac{115}{343}$. After maximizing $S_2\big(F^{\otimes 2}\,U_{49}(\phi_1,\,\phi_2)\big)$ over the parameters we find that for all $\phi_1,\,\phi_2$

\begin{equation}
    \text{range}\big(S_2(U_{49})\big) \supset \qty[0.042,\,\frac{115}{343}].
\end{equation}
Surprisingly, attempts to improve these bounds using simulated annealing techniques do not show any improvement over the simple approach we have used. The table of estimated coherence ranges obtained for other entropies and comparison with permutations is presented in Appendix \ref{app:coh_measure}.

The obtained solution provides also an explicit recipe to generate $\text{AME}(4,7)$ states of four systems of dimension $7$ by $|\text{AME}(4,7)\ra = \sum_{i,j = 1}^7 |i,j\ra \otimes ~ U_{49} |i,j\ra~.$

\begin{figure}[h!]
    \centering
    \includegraphics[width=\linewidth]{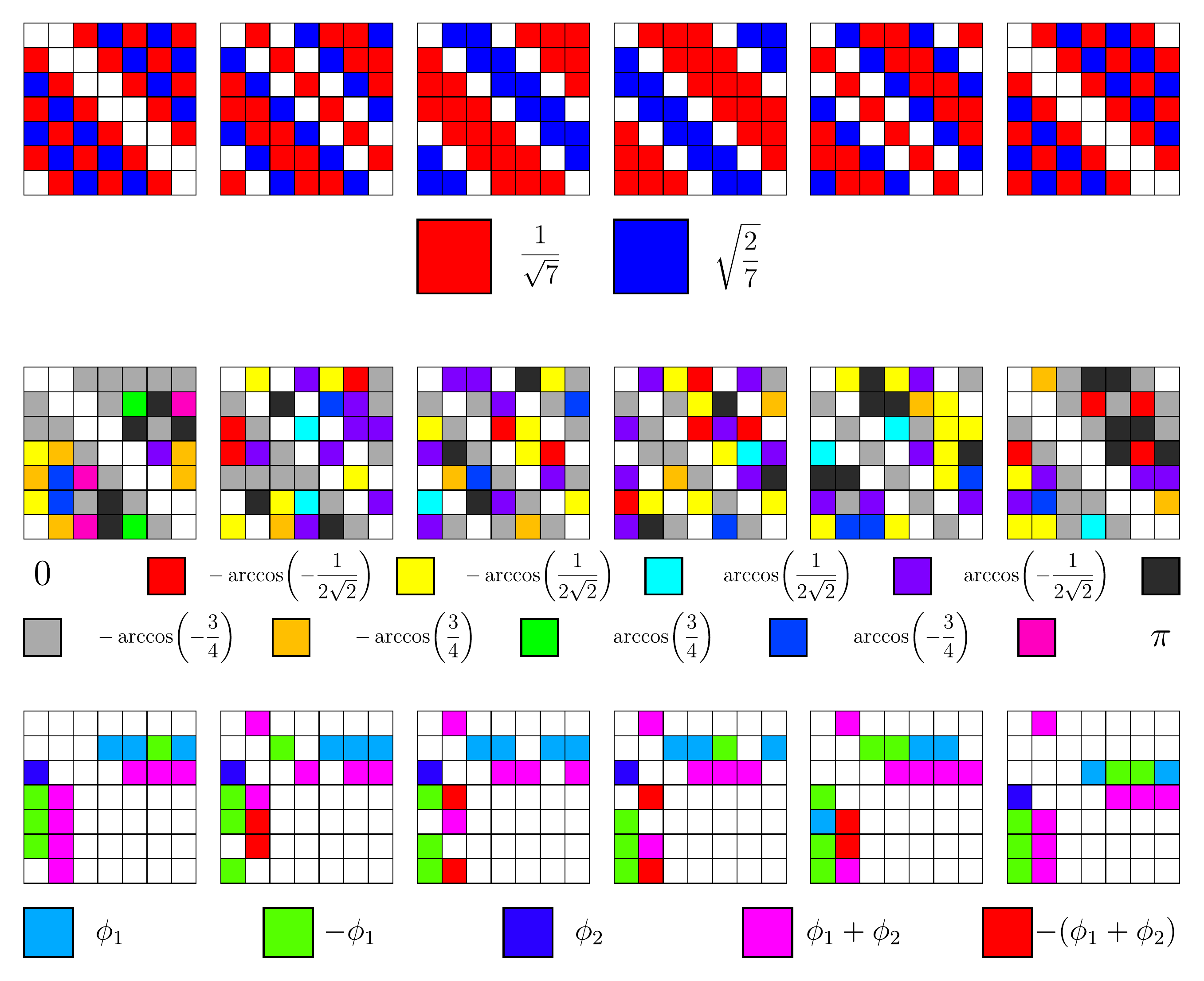}
    \caption{Visual representations of the bases $\{|a_{(k,l)}\ra\}$ generating the family $U_{49}$ of 2-unitary gates of dimension $d=7$ by equation \eqref{U_form1}. The first basis $k = 1$ is omitted since we set it to the computational basis by local transformation. Top row represents amplitudes $|a_{kli}|$ of the states ($k = 2,\cdots,7$). Middle row shows constant contributions to the phases of the form $\exp(i\phi_{kli})$. Last row represents the distribution of free contributions to the phases. Colours labeling the values are displayed in a vertical manner. White spaces in the first two rows represent zero elements and the elements without the free phase in the last row.}
    \label{fig:dimension_7}
\end{figure}

\subsection{Dimension 9}

In case of dimension $9$ let us consider unitary matrices $V_k$ of size $9 \times 9$ composed from consequent basis vectors $\ket{a_{k,l}}$,
\begin{equation}
    V_k = \qty(\ket{a_{k,1}},\hdots,\ket{a_{k,9}}).
\end{equation}
Let us define a cyclic permutation on blocks,
\begin{equation}
    \pi_{blocks} = (147)(258)(369)
\end{equation}
which we use in defining
\begin{equation}
    V_{k+3m} = P_{\pi_{blocks}}^m V_k~,
\end{equation}
where $P_{\pi}$ is the matrix representation of permutation $\pi$. As before, we may fix the first basis to be given by the computational basis,
\begin{equation}
    V_1 = \mathbb{I}~.
\end{equation}
The remaining two bases $V_2$ and $V_3$ are parameterized by two independent cyclic unitary matrices $B_2,\,B_3$,
\begin{equation}
\label{eq_9_blocks}    
    B_k = \mqty(
        a_k & b_k e^{i\phi_k} & c_k e^{i\theta_k} \\
        c_k e^{i\theta_k} & a_k & b_k e^{i\phi_k} \\
        b_k e^{i\phi_k} & c_i e^{i\theta_k} & a_k).
\end{equation}
with $a_k^2+b_k^2+c_k^2 = 1$ and $\phi_k$ and $\theta_k$ such that $B_kB_k^\dagger = \mathbb{I}$.

Then, we define for $k = 2, 3$
\begin{align}
\label{V9_form}
    V_k = P_{\pi_k}\qty(\bigoplus^3 B_{k}) P_{\sigma}~,
\end{align}
with permutations
\begin{align}
    \sigma = (24)(37)(68) 
\end{align}
and
\begin{align}
    \pi_2 = (123456789)\circ\sigma~, &&
    \pi_3 = (135792468)\circ\sigma~,
\end{align}
which, overall, yields the structure of entries as in Fig.~\ref{fig:dimension_9}, where each distinct number is marked by a different colour.
The 2-unitary matrix $U_{81}$ can be reconstructed using \eqref{U_form1} as:
\begin{equation}
[U_{81}]_{ki,jl} = A_{klj} [V_k]_{li}
\end{equation}
with $V_k$ as described in equation \eqref{V9_form} and $A_{klj} = \delta_{k,l\oplus j}$.
Note that the entire matrix $U_{81}$ has 4 free parameters, 
since unitarity conditions for $B_2$ (and $B_3$) reduce the number of their free parameters to $2$ (and $2$),
which in turn guarantees 2-unitarity conditions on $U_{81}$ constructed from the bases $V_k$ according to the recipe \eqref{U_form1}.

\begin{figure}[h]
    \centering
    \includegraphics[width=.6\linewidth]{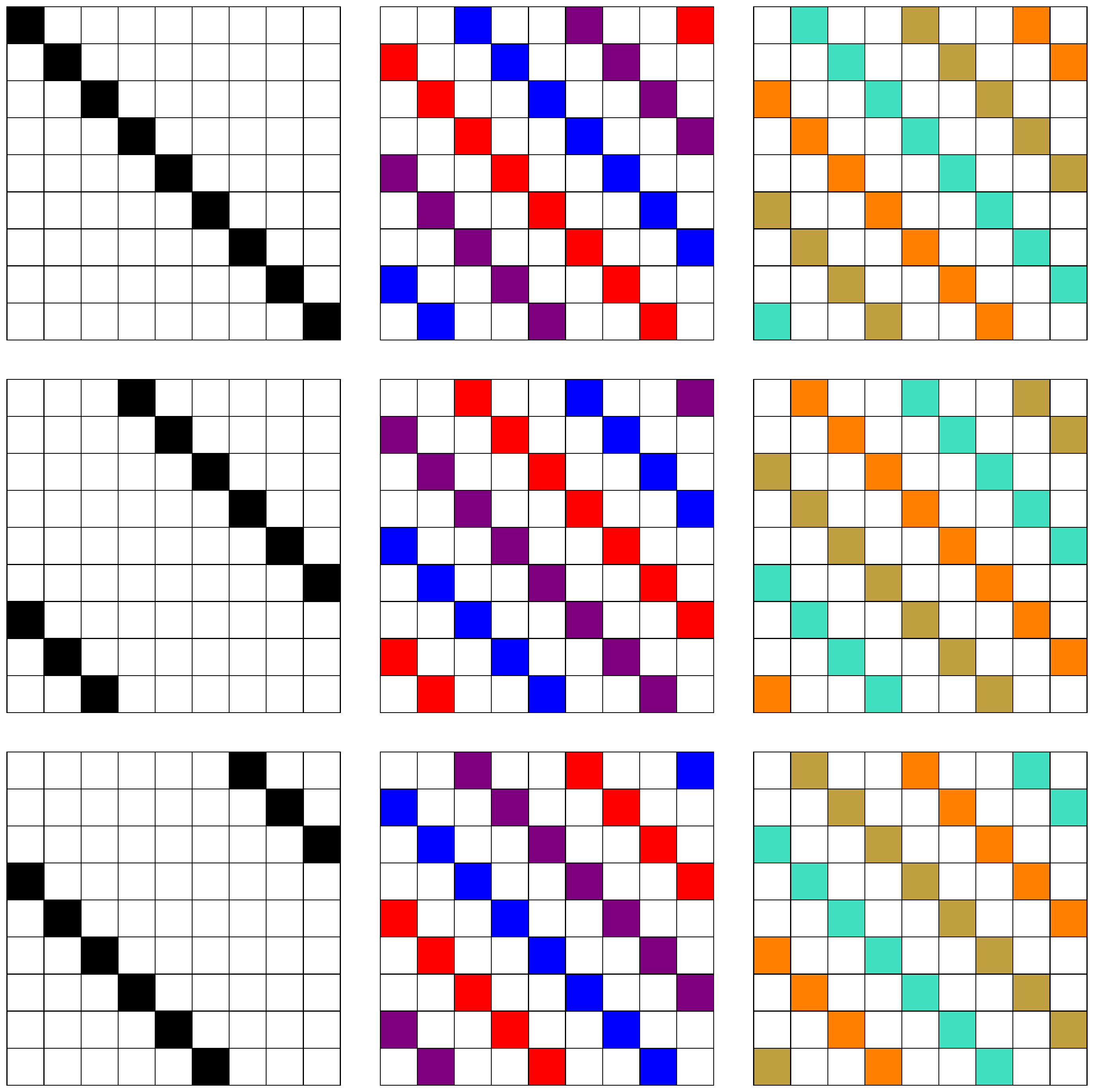}
    \caption{Visual representation of the basis $\{|a_{kl}\ra\}$ generating the family $U_{81}$ of 2-unitary gates with local dimension $d=9$, defined by equation \eqref{U_form1}. Triples of entries in the colours blue, violet, red (and cyan, brown, orange) correspond to entries of cyclic unistochastic matrices $B_2$ (and $B_3$) from \eqref{eq_9_blocks}, each colour representing one number.}
    \label{fig:dimension_9}
\end{figure}

\begin{figure*}[t]
    \centering
    \includegraphics[width=1\linewidth]{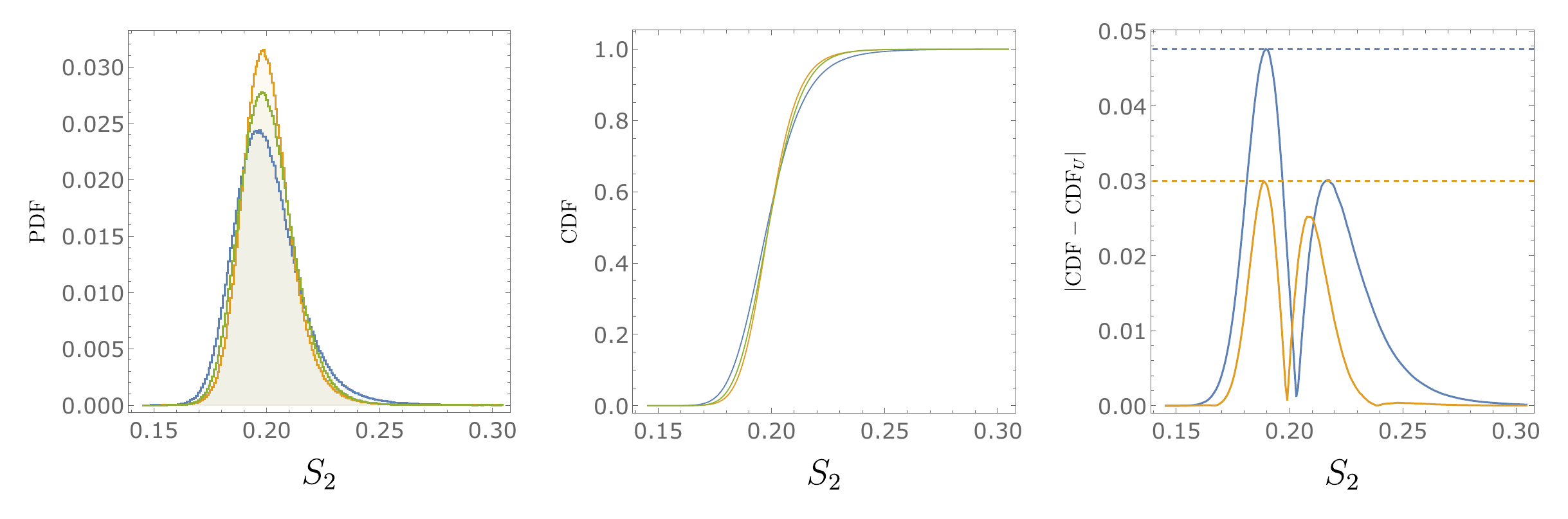}
    \caption{a) Density distributions of linear entropy for outputs of $U_{81}$ (green) and two nonequivalent AME permutations $P_{81}$ (orange, blue) action on separable states. b) Cumulative of those distributions, c) Absolute differences between cumulatives for $U_{81}$ and AME permutations $P_{81}$, with respective maximal values of 0.048 and 0.03, corresponding to $p$-values of distinguishability in Kolmogorov-Smirnov test with $2\cdot10^6$ samples equal to $10^{-1965}$ and $10^{-778}$, respectively.
    }
    \label{fig:dimension_9_hists}
\end{figure*}

Notice, that for limit values of parameters the matrix $U_{81}$ degenerates to permutation matrices. Thus it can be interpreted as a generalization and extension of 2-unitary permutations into continuous families of 2-unitary matrices.

Using the invariants $I_{\sigma, \tau, \rho, \lambda}(U_{81})$ we have not been able to demonstrate local inequivalence of $U_{81}$ from the 2-unitary permutations, and thus we resorted to a statistical approach used in~\cite{new_AME_16}. Histograms of generated entanglement (see Fig.~\ref{fig:dimension_9_hists}) show that the family $U_{81}$ is indeed locally distinct from permutations. We verify this quantitatively by using the two-sample Kolmogorov-Smirnov test with $2\cdot10^6$ samples per distribution, which yields the confidence level of at most $p = 10^{-778}$, implying that the sample obtained from $U_{81}$ is different from the standard construction for the 2-unitary permutations $P_{81}$ with probability $1-p$. 

We evaluate the coherence measures $S_2(U_{81})$ and $S_2\big(F_9^{\otimes 2}U_{81}\big)$ to find an inner bound on the coherence range of any member of the family as
{\footnotesize
\begin{equation}
    \begin{aligned}
        \text{range}\big(S_2(U_{81})\big) \supset \left[\frac{1}{243}\right.&\qty(1+a_1^4+b_1^4+c_1^4 +a_2^4+b_2^4+c_2^4), \\ 
\frac{1}{3}&\qty(1+a_1^4+b_1^4+c_1^4 +a_2^4+b_2^4+c_2^4)\biggr].
    \end{aligned}
\end{equation}}

The table of estimated coherence ranges obtained for other entropies and comparison with permutations is presented in Appendix \ref{app:coh_measure}.

The obtained solution provides a recipe to generate the absolutely maximally entangled state $\text{AME}(4,9)$  by the same token as in the dimension $d = 7$.

\section{Multipartite Convolutional Channels}\label{sec:multistoch}

Finally, we present the extension of our results for the multi-partite systems. To do so, we start by generalizing the notion of tristochasticity into multi-stochasticity. On the classical level, such an abstraction is quite natural.

\begin{defi}
    A tensor $A$ is called \textit{multistochastic} if $A_{i_1,\hdots,i_j,\hdots,i_n} \geq 0$ and $\sum_{i_j} A_{i_1,\hdots,i_j,\hdots,i_n} = 1$ for \mbox{$j\in(1,\hdots,n)$}.
\end{defi}

On the quantum level, the notion of multistochasticity is generalised
in the same spirit but involved formulas are slightly more convoluted.

\rd{\begin{defi}\label{multi_q_def}\cite{bistron2023tristochastic}
Channel $\Phi_D: \Omega_{d^{m-1}} \to \Omega_d$ defined by dynamical matrix $D_{j_1, \cdots, j_m}^{i_1, \cdots, i_m}$ by
{
\begin{equation}
\Phi_D[\rho_2 \otimes \cdots \otimes \rho_{m}] = \Tr_{2, \cdots,m}[D (\mathbb{I} \otimes \rho_2^\top \otimes \cdots \otimes \rho_{m}^\top)]~,
\end{equation}}
is called an $m$-stochastic channel, if for any sequence of density matrices $\{\rho_2, \cdots, \rho_{m}\}$ and any index \mbox{$k \in \{1, \cdots, m\}$} the map:
\begin{equation}\label{eq_multi_D2}
\Tr_{1,\cdots,k-1,k+1,\cdots,m-1 }\;[D(\underbrace{\rho_2^\top \otimes \cdots \otimes }_{k-1\text{ elements}}\mathbb{I} \otimes  \cdots \otimes \rho_{m}^\top) ]~,
\end{equation}
also forms a legitimate quantum channel.
\end{defi}}

While constructing multipartite convolutional channels, we also refer to the generalizations of Latin squares -- Latin cubes and Latin hypercubes~\cite{LatinHypercubes}.

\begin{defi}\label{def:multo_perm_coch}
A Latin hypercube $L_{i_2, \cdots, i_{m}}$ of dimension $d$ is a tensor with entries from the set $[d] = \{1,\cdots,d\}$, such that every hypercolumn $\qty{L_{i_2,\hdots,i_k,\hdots,i_m}}_{i_k=1}^d$ contains all the elements from the set $[d]$.
\end{defi}

We observe
that by fixing all the indices in a Latin hypercube $L_{i_2, \cdots, i_{m}}$ except two, one obtains a Latin square. We call any such square a \textit{Latin subsquare}. Using this observation one defines orthogonal Latin hypercubes as in~\cite{LatinHypercubes}. 

\begin{defi}
Two Latin hypercubes are orthogonal if each corresponding pair of Latin subsquares are orthogonal.   
\end{defi}

The order of indices \rd{$i_2,\hdots,i_m$} introduced above provides a natural translation between Latin hypercubes (and squares) and multistochastic permutation tensors and \textit{vice versa}.
Notice that 
\begin{equation}
\label{Latin-perm}
    A_{i_1, i_2, \cdots i_m} = \delta_{i_1,L_{i_2, \cdots i_m}}
\end{equation}
satisfies all the necessary conditions of the permutation tensor, as the sum over any of the indices $i_1, \cdots i_m$ on the right-hand side of \eqref{Latin-perm} gives one. On the other hand, one may define Latin hypercube by
\begin{equation}
    L_{i_2,\cdots, i_m} = i_1 \text{~such that~} A_{i_1,\cdots,i_m} = 1~.
\end{equation}
Once again, the defining property of Latin hypercubes is satisfied, because for any fixed values of $i_2, \cdots\rd{i_{k-1},i_{k+1},\cdots}, i_m$ and two different $i_k\neq i_k'$ if
\begin{equation*}
L_{i_2,\cdots,i_k,\cdots i_m} = L_{i_2,\cdots, i_k',\cdots, i_m} = i_1
\end{equation*}
then 
\begin{equation*}
 A_{i_1,\cdots,i_k,\cdots,i_m} =  A_{i_1,\cdots,i_k',\cdots,i_m} = 1.   
\end{equation*}
\rd{The above implies that the sum over the hypercolumn defined by the set of fixed indices given above yields 2, in contradiction with multistochasticity of $A$.}

\subsection{Coherification of multi--stochastic permutation tensors}\label{sec:multi_perm_coch}

Now we are equipped with all the necessary tools to study multipartite convolutional channels $\Phi_A$ associated with $m$-stochastic permutation tensors $A$. As it turns out, those channels can also be realized as a unitary

\begin{equation}
\label{multistoch_U}
U_{i_1 j_1 \;j_2 \;\cdots \; j_{m-2}}^{\;i_2 \;i_3 \;i_4 \;\cdots\; i_m\;} = A_{i_1 i_2 \cdots i_m} \ket{a_{(i_1 i_2 \cdots i_{m-1})}}_{j_1  j_2 \cdots j_{m-2}}~,
\end{equation}
where for any $i_1$ the vectors $\qty{|a_{(i_1;\;i_2 \cdots i_{m-1})}\ra}_{i_2,\cdots,i_{m-1} = 1}^d$ form a $d^{m-2}$ dimensional basis; 
followed by partial trace on all subsystems except the first one. 
Consult Appendix~\ref{app:milti_coch} for detailed derivation and examples.

Dynamical matrix $D$ of the channel $\Phi_A$ takes a form analogous to \eqref{D_tristoch1}
\begin{equation*}
D_{i_1'\;i_2' \; \;\cdots\;}^{\;i_1\;i_2 \;\cdots\;} =  A_{i_1 i_2 \cdots}   A_{i_1' i_2' \cdots} \la a_{(i_1 i_2 \cdots)}|a_{(i_1' i_2' \cdots )} \ra ~.
\end{equation*}
If one defines
\begin{equation}
|\vb{a}_{(i_1 i_2 \cdots i_m)} \ra := A_{i_1 i_2 \cdots i_m} |a_{(i_1 i_2 \cdots i_{m-1})} \ra~,
\end{equation}
the condition for unitarity of $U$ \eqref{multistoch_U} can be expressed in the spirit of formulae \eqref{multi_tri1} as

\begin{equation}
\label{multi_multi1}
\forall_{i_1} \forall_{i_2 i_2'\cdots i_m i_m'}~~\la \vb{a}_{(i_1 i_2 \cdots i_m)}|\vb{a}_{(i_1 i_2' \cdots i_m')}\ra = \delta_{i_2 i_2'}\cdots \delta_{i_m i_m'}~.
\end{equation}
and the complementary conditions read,
\begin{equation}
\label{multi_multi2}
\begin{aligned}
& \forall_{i_2} \forall_{i_1 i_1'\cdots i_m i_m'}~~\la \vb{a}_{(i_1 i_2 \cdots i_m)}|\vb{a}_{(i_1' i_2 \cdots i_m')}\ra = \delta_{i_1 i_1'}\cdots \delta_{i_m i_m'}, \\
& \forall_{i_3 = i'_3} \forall_{i_1 i_1'\cdots i_m i_m'}~~\la \vb{a}_{(i_1 i_2 \cdots i_m)}|\vb{a}_{(i_1' i_2' \cdots i_m')}\ra = \delta_{i_1 i_1'}\cdots \delta_{i_m i_m'}, \\
& \cdots
\end{aligned}
\end{equation}

These conditions let us define unitary channels in all the other choices of input and output spaces, analogous as in \eqref{new_tri_u}, hence they correspond to (quantum) $m$-stochasticity of $\Phi_A$.

One may expect that conditions \eqref{multi_multi1}, \eqref{multi_multi2} would be sufficient to guarantee also maximal entangling power of the multipartite unitary channel $U$, similarly as in the case of convolutional channels. However, this is not the case. While considering multipartite entangling power \cite{linowski2019entangling}, see eq. \eqref{ep_multi_def}, one must consider all the bipartitions for both input and output indices of $U$: $p|q$ and $x|y$, see \eqref{Epq_def} in Appendix \ref{sec:multi_eq}. On the other hand, in the equations  \eqref{multi_multi1}, \eqref{multi_multi2} all except one output indices of $U$: $j_1 \cdots j_{m-2}$ are always together. Thus quantum multi-stochasticity is a weaker demand. 

For an extended discussion of entangling power in the context of multipartite channels we encourage the reader to consult Appendix \ref{sec:multi_eq}. \rd{Finally in Appendix \ref{app:max_coch} we generalize the above discussion for arbitrary tristochastic tensors,  which have no explicit formulae for optimal coherification.}

\subsection{Latin (hyper)cubes and their connection to maximal $e_p$}

Finally, we present an example of a multipartite convolutional channel associated with $m$-stochastic permutation tensor, which is \textit{both} (quantum) multi-stochastic and has a maximal multipartite entangling power, demonstrating that our framework generalizes previously known examples.

Let $A_{i_1 \cdots i_m}$ be a permutation tensor of interest, $L_{i_2\cdots i_m}^{(1)}$ corresponding Latin hypercube and $L_{i_2\cdots i_m}^{(2)}, \cdots, L_{i_2\cdots i_m}^{(m-1)}$ be $m-2$ Latin hypercubes such that all Latin hypercubes $\{L^{(i)}\}$ are mutually orthogonal. Then the multipartite unitary $U$ corresponding to channel  $\Phi_A$ has a form

\begin{equation}
\label{multi_u_Latin}
U = \sum_{i_2,\cdots i_m} |L_{i_2\cdots i_m}^{(1)} L_{i_2\cdots i_m}^{(2)} \cdots L_{i_2\cdots i_m}^{(m-1)} \ra \la i_2 \cdots i_{m}|~.
\end{equation}

Because Latin hypercubes ${L^{(i)}}$ are mutually orthogonal, by Theorem 5.12 from~\cite{LatinHypercubes}, construction \eqref{multi_u_Latin} gives a large permutation matrix, hence a unitary matrix.

To argue the maximum entangling power of \eqref{multi_u_Latin} we use the fact that vectorised unitary matrix $\ket{U}$ is an AME state (see~\cite{ameMasterThesis} section 3.2) so all the partitions in \eqref{Epq_def} gives maximal possible contribution to entangling power.
Since $\ket{U}$ defined in \eqref{multi_u_Latin} is an AME state, a simple argument for the multi-stochasticity of $\Phi_A$ follows. The maximal entanglement of $|U\ra$ with respect to bipartition $i_2,\cdots,L_{i_2\cdots i_m}^{(1)},\cdots, i_m|i_k,L_{i_2,\cdots, i_m}^{(2)}\cdots L_{i_2,\cdots, i_m}^{(m-1)}$, guarantee that the matrix:
\begin{equation}
\sum_{i_2,\cdots i_m} |i_k L_{i_2\cdots i_m}^{(2)}  \cdots L_{i_2\cdots i_m}^{(m-1)} \ra \la i_2 \cdots L_{i_2\cdots i_m}^{(1)}  \cdots i_{m}|
\end{equation}
is unitary for any $k$.
In Appendix \ref{app:milti_coch_multi}, Theorem \ref{quantum_multi_thm} we present also an alternative the proof of multi stochasticity for $\Phi_A$ corresponding to unitary channel \eqref{multi_u_Latin}.

Although the existence of orthogonal Latin hypercubes is far less explored than for orthogonal Latin squares, some results are known. For example, thanks to Theorem~5.4 form~\cite{LatinHypercubes}, we are guaranteed that for $d$ being prime power and $2< m \leq d+1$ there exist at least $d-m+2$ mutually orthogonal Latin hypercubes of order $m-1$. This means that in the prime power dimension $d$ our construction is valid if $(d+1)/2 \geq m-1$.

\section{Outlook and conclusions}\label{sec:outlook}

Our work serves as a first step on a new trail for constructing highly entangling operations. 
In particular, we arrive at novel families of 2-unitary matrices with free non-local parameters beyond simple phasing. 

First, we considered the entire set of convolutional channels, to show the full range of possibilities for such construction. 
Using the framework of coherification of permutation tensors, we introduced new continuous families of 2-unitary matrices in dimensions $7\times 7$ and $9\times 9$, and emphasize their particular properties. 
\rd{Furthermore, we proposed a new measure of coherence for unitary operations, based on the range of Rényi entropies generated from computational basis inputs. This measure captures the ability of a unitary to generate nontrivial coherence under arbitrary local pre- and post-processing. Our approach builds on well-established coherence monotones for quantum states and extends earlier entropy-based methods from single-value metrics to the entire spectrum of Rényi entropies over local transformations.}

Thus we placed the first steps towards development of the theory of 2-unitary channels, which will allow for their parametric optimization for specific tasks. It is crucial to stress at this point that the introduced families are exemplary and the introduced framework
is not limited to them. 
We emphasise that 2-unitary matrices based on the construction introduced in this work are not equivalent to either the standard orthogonal Latin squares construction or other non-standard approaches. We demonstrated the former explicitly, using invariants and statistical methods~\cite{36_officers_of_Karol, rather2023biunimodular}. The latter can be found by noticing either a mismatch between the block structure~\cite{36_officers_of_Karol} of the solutions or the lack of continuous non-local parameterization~\cite{Suahil_invariants}. \rd{Last but not least, in Appendix \ref{app:max_coch}, we show that each 2-unitary matrix corresponds to a tristochastic tensor, for which it can be considered a maximal coherification -- a hint at a deeper link akin to unitary and bistochastic matrices.} 

Convolutional channels were based primarily on tristochastic tensors\rd{, thus} resulting in bipartite unitary matrices. However, in the final Section \ref{sec:multistoch}, we also generalize \rd{our approach} for multistochastic permutation tensors giving multipartite unitaries with large entangling and disentangling capacities.

Possible application of our work, beyond the new frontier of the search for perfect tensors, might be its implementation into the recently emerging field of quantum convolutional neural networks (qCNN)~\cite{QCNN1, QCNN2, QCNN3, Herrmann}. To fully translate the idea of convolutional neural network on the quantum framework, one has to replace classical states and operations with their quantum counterparts in a suitable way.
Notably, the convolution layers of quantum networks necessitate a quantum equivalent of the convolution and pooling operation. Such an operation should possess several desired properties:  (a) the ability to disentangle entangled states, converting non-local correlations into properties of local states; (b) nontrivial impact on computational states, leveraging quantum properties by introduction of coherence; (c) parametrizability, necessary to facilitate the training of convolutional layers. Given that entangling power is proportional to a less-known disentangling power~\cite{CLARISSE2007400}, the proposed framework of 2-unitary operations emerges as a strong candidate satisfying the above properties.

In Appendix \ref{sec:2x2} we present a simple case study, on the example of convolutional channels constructed from orthogonal Latin squares, to show its limited, nevertheless quite remarkable, capabilities in disentangling not only quantum states but the entire maximally entangled basis.

Our work prompts important and intriguing questions worth further investigation. First and foremost, it is tempting to try to generalize our findings into a universal recipe for continuous families of multi-unitary matrices in arbitrary dimension $d$. The dimensions $d = 2^n$ are of special interest due to possible applicability in quantum circuits. The next open problem is to construct a quantum circuit that corresponds to such channels, which is crucial for real-life applications. Finally, the issue of connecting the convolutional channels into larger networks has only been touched upon and requires further study for more general channels.

\medskip
\noindent
\rd{
\textit{Note added:}
During the publication process another approach to quantum convolution, particular useful in description and characterization of magic and stabilizer states, was put forward in \cite{Bu2023} and later extended in \cite{Bu2025} in the context of the so called quantum Fourier analysis. We note that this approach falls within our framework with Definition 8 from \cite{Bu2025} corresponding to \eqref{multistoch_U}, or Definition \ref{def:convololo_channel} for two-argument case, while the coherification property is explicitly proven in Proposition 12 from \cite{Bu2025}.}

\subsection*{Acknowledgments}

It is a pleasure to thank Wojciech Bruzda and Adam Burchardt for fruitful discussions. Moreover, we thank Grzegorz Rajchel-Mieldzioć, Arul Lakshminarayan, Suhail Rather and Michael Zwolak for valuable suggestions.
Financial support by NCN under the Quantera project no. 2021/03/Y/ST2/00193 and PRELUDIUM BIS no. DEC-2019/35/O/ST2/01049 is gratefully acknowledged.

\appendix

\section{Case study: 2-unitary from orthogonal Latin squares}\label{sec:2x2}

In this Appendix we aim to present previously neglected properties of $2$-unitary, and later $3$-unitary, matrices -- its disentangling power. More precisely we will demonstrate in the simple case study, that multi-unitary matrices, which fall in our framework of convolutional channels, can disentangle the entire basis of maximally entangled states into a separable one.
 
Let $A$ be a tristochastic permutation tensor of interest, $L$ a corresponding Latin square and $M$ Latin square orthogonal to $L$, then the maximally disentangling unitary matrix in the channel $\Phi_A$ \eqref{U_form1} can be constructed as
\begin{equation}
    P_{d^2} = \sum_{lj} |L_{lj}, M_{lj}\ra\la l,j|~, \tag{\ref{U_orto_Latin_s} revisited}
\end{equation}
with the same relation between Latin squares $L,M$ and the permutation tensor $A$ with vectors $|a_{kl}\ra$ as in  \eqref{ort_Latin_basis}.

Unitary matrix \eqref{U_orto_Latin_s} has maximal entangling power $e_p = 1$ and gate typicality $g_t = \frac{1}{2}$. 
Since orthogonal Latin squares exist in any dimension except $d = 2$ and $d=6$~\cite{bose_shrikhande_parker_1960},  construction \eqref{U_orto_Latin_s} is rather general. On the downside, it does not provide any free parameters except possible phases.

As the case study, which can be efficiently implemented in the modern quantum computer, we consider a channel $\Phi_A$ constructed via \eqref{U_orto_Latin_s} for two ququarts. Since each ququart can be interpreted as a pair of qubits, from the perspective of quantum hardware the channel $\Phi_A$ is a unitary $P_{16}$ acting on four qubits, followed by a partial trace on two of those.

In general, there is a large freedom in the construction of orthogonal Latin squares $L$ and $M$, which correspond to local gates. One may simultaneously permute rows and columns of these squares, which corresponds to local preprocessing $v_1\otimes v_2$ of the unitary channel $P_{16}$ \eqref{U_orto_Latin_s} or permute the symbols in the Latin squares, which corresponds to local postprocessing $v_1'\otimes v_2'$ of $P_{16}$, resulting in a locally equivalent channel of the form $(v_1\otimes v_2) P_{16} (v_1'\otimes v_2')$. To reduce the number of such "repetitive" channels, we fixed the first columns of $L$ and $M$ and the first row of $M$ to be $(1,2,3,4)$, which almost completely erases such degeneration. After all the eliminations associated with local preprocessing and postprocessing, the remaining pair  $(L, M)$ of orthogonal Latin squares in dimension $d = 4$ is given below
\begin{equation*}
L = \left(
\begin{array}{cccc}
   1&  2&  3&  4\\
   2&  1&  4&  3\\
   3&  4&  1&  2\\
   4&  3&  2&  1\\
\end{array}
\right) ~,~~
M = \left(
\begin{array}{cccc}
   1&  2&  3&  4 \\
   3&  4&  1&  2 \\
   4&  3&  2&  1 \\
   2&  1&  4&  3 \\
\end{array}
\right) ~.
\end{equation*}
This pair $(L, M)$ gives a unitary channel $P_{16}$  which can be viewed either as two ququart maximally (dis)entangling gate or, after decomposing the ququarts, a four qubit highly (dis)entangling gate. It can be implemented using a circuit of depth 11 using 18 nearest-neighbour $CNOT$ gates in the following way: 
\begin{widetext}
\vspace{-0.05 cm}
    \begin{equation}
    {\Qcircuit @C=1.0em @R=.7em {
        &\multigate{3}{P_{16}} & \qw & \raisebox{-4.8em}{=} &&\ctrl{1}&  \qw   & \targ   &  \qw    & \targ   &  \qw    &\ctrl{1}&
         \targ   &  \qw   &\ctrl{1}& \targ & \qw     \\
         &\ghost{P_{16}} & && \qw & \targ  &\ctrl{1}&\ctrl{-1}& \targ   &\ctrl{-1}& \targ   & \targ  &
        \ctrl{-1}&\ctrl{1}& \targ   &\ctrl{-1} & \qw \\
        &\ghost{P_{16}} & \qw &&&\ctrl{1}& \targ  & \targ   &\ctrl{-1}& \targ   &\ctrl{-1}&
         \targ   &\ctrl{1}& \targ  &\ctrl{1}& \targ& \qw     \\
        &\ghost{P_{16}} & \qw &&& \targ  &  \qw   &\ctrl{-1}&  \qw    &\ctrl{-1}&  \qw    &
        \ctrl{-1}& \targ  & \qw    & \targ  &\ctrl{-1}& \qw
        \gategroup{1}{1}{2}{1}{.7em}{--}
        \gategroup{3}{1}{4}{1}{.7em}{--}
        \gategroup{1}{3}{2}{3}{.7em}{--}
        \gategroup{3}{3}{4}{3}{.7em}{--}
        \gategroup{1}{5}{2}{5}{.7em}{--}
        \gategroup{3}{5}{4}{5}{.7em}{--}
        \gategroup{1}{17}{2}{17}{.7em}{--}
        \gategroup{3}{17}{4}{17}{.7em}{--}
    }}
\end{equation}
\vspace{-0.05 cm}
\end{widetext}
Notice that one left outer layer and two right outer layers can be  "pulled out" as local ququart pre- and postprocessing reducing the circuit to $12$ gates organized in 8 layers.
The matrix $P_{16}$ acts on two ququarts each represented by a pair of qubits encompassed by a dashed rectangle.
Up to our knowledge, this is the most efficient way to implement a 2-unitary matrix $P_{16}$ using only nearest neighbour gates in linear architecture. 

For the sake of completeness, we recall
that $P_{16}$ is a permutation matrix of order $16$, thus all the vectors from the computational basis are mapped onto each other, so the results are separable.

On the other hand, it is intriguing that there exists a basis of maximally entangled states of two ququarts, for which all vectors are mapped by $P_{16}$ onto separable states. Therefore the action of $\Phi_A$ on all the vectors from this basis gives a set of pure states, which overlap with the ququart basis.

To present this basis and discuss more of its profitable properties let us first introduce a suitable notation. Let
\begin{equation*}
|\Psi_{\pm}\ra = \frac{|00\ra \pm |11\ra}{\sqrt{2}} ~,~~~ |\Xi_{\pm}\ra = \frac{|01\ra \pm |10\ra}{\sqrt{2}}~,
\end{equation*}
denote the Bell states, entangling the first or second qubits from each ququart. Then the discussed basis takes the form
{\footnotesize
\begin{equation}
\label{monster_basis}
\left\{
\begin{matrix}
|\Psi_{+}\ra \otimes |\Psi_{+}\ra, & \wm|\Psi_{+}\ra \otimes |\Psi_{-}\ra, & \wm|\Psi_{-}\ra \otimes |\Psi_{+}\ra, & \wm|\Psi_{-}\ra \otimes |\Psi_{-}\ra,\\

|\Psi_{+}\ra \otimes |\Xi_{+}\ra, & \wm|\Psi_{+}\ra \otimes |\Xi_{-}\ra, & -|\Psi_{-}\ra \otimes |\Xi_{+}\ra, & -|\Psi_{-}\ra \otimes |\Xi_{-}\ra,\\

|\Xi_{+}\ra \otimes |\Psi_{+}\ra, & -|\Xi_{+}\ra \otimes |\Psi_{-}\ra, & -|\Xi_{-}\ra \otimes |\Psi_{+}\ra, & \wm|\Xi_{-}\ra \otimes |\Psi_{-}\ra,\\

|\Xi_{+}\ra \otimes |\Xi_{+}\ra, & -|\Xi_{+}\ra \otimes |\Xi_{-}\ra, & \wm|\Xi_{-}\ra \otimes |\Xi_{+}\ra, & -|\Xi_{-}\ra \otimes |\Xi_{-}\ra
\end{matrix}
\right\}.
\end{equation}}

The vectors from consecutive rows of basis \eqref{monster_basis} map under $P_{16}$ onto basis vectors
\begin{equation*}
\{|00\ra~,~~ |01\ra~,~~ |10\ra ~,~~ |11\ra  \}  
\end{equation*}
on the first ququart and the vectors from consecutive columns map onto basis vectors
\begin{equation*}
\left\{
\begin{matrix}
(|00\ra + |01\ra + |10\ra + |11\ra)/2, \\
(|00\ra + |01\ra - |10\ra - |11\ra)/2, \\
(|00\ra - |01\ra - |10\ra + |11\ra)/2, \\
(|00\ra - |01\ra + |10\ra - |11\ra)/2~ \\
\end{matrix} \right\}
\end{equation*}
on the second ququart.
Thus the choice of type of Bell states ($|\Psi_{\pm}\ra$ or $|\Xi_{\pm}\ra$) determines the result on the first ququart, and the signs chosen in them (eg. $|\Psi_{+}\ra$ or $|\Psi_{-}\ra$)  determine the result on the second ququart.

Due to such an elegant mapping of basis vectors \eqref{monster_basis} under $P_{16}$ we can say even more about the action of $P_{16}$ on maximally entangled states of two ququarts.  For example, if one constructs such a state as a superposition of vectors from one row (or column) from basis \eqref{monster_basis}, then after the action of $P_{16}$  all these basis vectors will map on the same pure state on the first (second) ququart. Therefore the action of $P_{16}$ on such superposition also gives a separable state, hence action of $\Phi_A$ gives pure output.

Generalizing this property on the pairs or triples of columns and row form \eqref{monster_basis} one obtains the following result.
\begin{theo}\label{theo_square4_map}
Let $|\psi\ra$ be any state of two ququarts, whose decomposition in the basis \eqref{monster_basis} employ the vectors from $m$ rows and $n$ columns of \eqref{monster_basis}. Then the maximal number of nonzero eigenvalues of $\Phi_A(|\psi\ra\la\psi|)$ is equal to $min(m,n)$.
\end{theo}

The channel $\Phi_A$ may be considered as a prototype for a building block in quantum convectional neural networks (qCNN). One can create qCNN acting on several ququarts, by stacking the discussed channel $\Phi_A$ parallely or sequentially, with suitable single-ququart gates along the way. To find a basis of entangled states transformed by such circuits into pure computational states, one only needs to iteratively combine the basis \eqref{monster_basis} with itself in an appropriate way.

\subsection{3-unitary from orthogonal Latin cubes}\label{app:3x2}

The construction presented above may be generalized into
multi-stochastic quantum channels. As an example, we briefly discuss the channel obtained from   \eqref{multi_u_Latin} using three orthogonal Latin cubes of dimension $d = 4$ presented in Fig.~\ref{hyper_orto_fig}.

\begin{figure}[ht]
        \includegraphics[width=\linewidth]{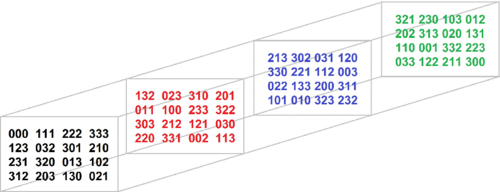}
    \caption{\label{hyper_orto_fig}
Three orthogonal Latin cubes of dimension $d = 4$ corresponding to 3-unitary matrix of size $4^3 = 64$ via equation \eqref{multi_u_Latin} with $m = 3$. Figure borrowed from~\cite{PhysRevA.92.032316} }
\end{figure}

Unitary $U$ \eqref{multi_u_Latin} from this channel acts on three quqarts, which we interpret as three pairs of qubits, same as previously.  
In this case, there also exists a maximally entangled basis, all of which elements are mapped into fully separable states of three ququarts. Therefore the action of $\Phi_A$ on those basis vectors gives pure states overlapping with the basis on quqart.

To present this basis let us entangle the first qubits from all ququarts and the second qubits from all quqarts by $GHZ$ states $|GHZ_{\pm}^i\ra$: 
\begin{equation}
\begin{aligned}
&|GHZ_{\pm}^1\ra = \frac{|000\ra \pm |111\ra}{\sqrt{2}} ~,~~ |GHZ_{\pm}^2\ra = \frac{|001\ra \pm |110\ra}{\sqrt{2}}~, \\
&|GHZ_{\pm}^3\ra = \frac{|010\ra \pm |101\ra}{\sqrt{2}} ~,~~ |GHZ_{\pm}^4\ra = \frac{|011\ra \pm |100\ra}{\sqrt{2}}~, \\
\end{aligned}
\end{equation}
Then the abovementioned basis has a form $\{|GHZ_{\pm}^i\ra \otimes |GHZ_{\pm}^j\ra \}$, where the indices $i,j$ and both sights $\pm$ are independent.
Moreover, after appropriate multiplication by $\pm 1$ of the basis vectors $\{|GHZ_{\pm}^i\ra \otimes |GHZ_{\pm}^j\ra \}$, one could repeat the above discussion, together with the analogue of Theorem \ref{theo_square4_map}, but this time on the three ququarts.

\section{Measure of coherence of a unitary operation}\label{app:coh_measure}

It is easy to measure the coherence of any pure state with respect to any given basis by considering the entropic properties of the resulting probability distribution.
Let us take a state $|\psi\ra$. Its coherence with respect to the
basis defined by a unitary matrix $U$ is given by \eqref{eq_whatever1}. 
For $\alpha \in \qty{0, 2, \infty}$ exponentials of these entropies, presented in equation \eqref{eq_whatever2},
turn out to have simple interpretations. In particular, $S_0$ counts the nonzero elements of $\ket{\psi}$ and $S_\infty$ is equal only to the absolute value of the largest element of $\ket{\psi}$. Finally, $S_2$ is closely connected to the linear entropy, often used in the context of entanglement. 

Based on the above, we may define measures for coherence of a unitary matrix $U$ based on average (or total) coherence generated  on the computational basis,

\begin{align}
    H_\alpha(U) = \frac{1}{D}\sum_{j=1}^D H_\alpha(\ket{j};U)~, \\
    S_\alpha(U) = \frac{1}{D}\sum_{j=1}^D S_\alpha(\ket{j};U)~,
\end{align}
where, again, the simple interpretation of $S_0
$ and $S_\infty$, is the average number of elements per vector and average maximal element. In this case, we have two apparent degrees of freedom to introduce -- freedom to change the measurement basis $\ket{j}$ to $W\ket{j}$, and to rotate the operation $U$ to $VUV^\dagger$. This yields the following expressions

\begin{align}
    H_\alpha(U;W,V) = \frac{1}{D}\sum_{j=1}^D H_\alpha(W\ket{j};VUV^\dagger), \\
    S_\alpha(U;W,V) = \frac{1}{D}\sum_{j=1}^D S_\alpha(W\ket{j};VUV^\dagger).
\end{align}
One can easily see that we may write explicitly
\begin{align}
\hspace{-0.15 cm}
    H_\alpha(U;W,V) &=\! \frac{1}{D(1-\alpha)}\!\sum_{j=1}^D \log(\sum_{i=1}^D \abs{\!\mel{i}{VUV'}{j}}^{2\alpha}\!), \\
    S_\alpha(U;W,V) &=\! \frac{1}{D}\sum_{i,j=1}^D \abs{\!\mel{i}{VUV'}{j}}^{2\alpha},
\end{align}
with $V' = V^\dagger W$. The quantity $S_\alpha(U;V,W)$ takes a particularly elegant form, reminiscent of the Welch bounds~\cite{datta2012geometry}. Using these bounds we find for $\alpha>1$

\begin{equation}
    S_\alpha(U;\mathbb{I},V)\;\geq\;\frac{D}{\binom{D+\alpha-1}{\alpha}},
\end{equation}
However, this bound is far from saturable, as for $\alpha = 2$ one would need at least $D^2$ vectors.

Such measures would be rendered meaningless given full freedom of basis choice -- every unitary can be equivalent to a diagonal matrix, or a Fourier matrix by a proper choice of $V$ alone, thus reaching minimal and maximal values, respectively. However, in realistic settings, we will usually be dealing with partial freedom. 

For instance, it is natural to assume that we deal with a bipartite system, $D = d^2$, and to restrict our attention to local bases, $V,W\in U(d)\otimes U(d)$. Then one may consider the possible range of entropies achievable,
\begin{widetext}
\begin{equation}
    {\operatorname{range}}(S_\alpha(U)) = \qty{\min_{V, W \in U(d)^{\otimes 2} }S_\alpha(U;W,V),\max_{V, W \in U(d)^{\otimes 2}} S_\alpha(U;W,V)}~,
\end{equation}
\end{widetext}
corresponding directly to formula \eqref{range_def1} by relation \mbox{$V' = V^\dagger W$}, \bl{with extrema taken once again over tensor products of local unitary operators from $U(d)$ i.e. $V = v_1\otimes v_2$, $W = w_1 \otimes w_2$.}

Equipped with these, we may start asking questions about possible ranges for different operators. One certain thing one can say is that if $U = U_A\otimes U_B$, then we cover the entire possible range for a given entropic measure, for example

\begin{equation}
    \operatorname{range}(S_0(U_A\otimes U_B)) = \qty{1,d^2}~,
\end{equation}
which are values for permutation and Hadamard matrices, respectively. 
The same holds also if $U = P$ is a permutation matrix, in particular, constructed from two orthogonal Latin squares

\begin{equation}
\label{P_coch_range}
    \operatorname{range}(S_0(P)) = \qty{1,d^2}~,
\end{equation}
which are values for a permutation ($V = W = \mathbb{I}$) and Hadamard matrices ($W = H_d \otimes H_d$, $V = \mathbb{I}$), respectively.
Those two examples are subcases of general observation:
\begin{lem}
Bipartite unitary matrix $U$ have a maximal coherence range if and only if it is locally equivalent to a permutation matrix.  
\end{lem}

This proceeds from the fact that to achieve both minimal and maximal coherences $S_\alpha$, the bipartite matrix must be locally equivalent to both a permutation matrix when all rows of $U$ are vectors in the computational basis, and the Hadamard matrix when are rows of $U$ are unbiased with respect to computational basis. Moreover, the latter stream from the former as presented in equation \eqref{P_coch_range}.

Using the above observation we may state one more property of the coherence range.

\begin{theo}
For any $S_\alpha$, the set of bipartite unitary matrices with maximal coherence range is a disconnected set of measure zero w.r.t Haar measure.
\end{theo}

\begin{proof}
The set of permutation matrices of size $d^2 \times d^2$ is a finite, disjoin subset of bipartite unitary matrices. Moreover the smooth mapping on entangling power $e_p$ -- gate typicality $g_t$ plane, preservers those properties~\cite{permutatrion_entangling_power}. On the other hand, the allowed values of entangling power and gate typicality for bipartite matrices form a non-degenerated area. Thus the set of unitary matrices locally equivalent to a permutation is both disjoint, because any path connecting two permutations with different $e_p$ cannot consist of permutations, and have measure zero.
\end{proof}

Due to the above, this measure doesn't suit monotone for a potential resource theory. The underlying free set would be non-convex and disconnected, which would impede the application of almost all known tools from the resource-theoretic field.

For a generic $U$ 
coherence
range is not trivial, especially the minimal value points towards the non-vanishing coherence of the matrix.
For example, the 2-unitary matrix for local dimension $d = 6$ obtained in~\cite{36_officers_of_Karol} cannot obtain the limit value of $S_\alpha$ corresponding to permutation, since there are no permutation 2-unitary matrices in local dimension $6$. 

Our solution $U_{49}$,which also locally inequivalent to a permutation, with $\phi_1=\phi_2=0$ yields 

{\footnotesize
\begin{align*}
    \qty[\frac{4255 -18\sqrt{2}}{117649}, \frac{115}{343}] 
    \subset ~
    \operatorname{range}(S_2(U_{49}(0,0)))  \subset \qty[\frac{1}{49},1] ~.
\end{align*}}
While in general case the minimum of $S_2$ is given by
\begin{widetext}
{\footnotesize
    \begin{align*}
    \min_{V, W\in U(d)^{\otimes 2}} \!\!\!\!\! S_2(U_{49}(\phi_1,\phi_2)) = \frac{1}{117649}&\left[4260-6 \sqrt{ 2} -3 \sqrt{7} \sin (\phi_1-\phi_2)-3 \sqrt{7} \sin (\phi_1+\phi_2)-3 \sqrt{14} \sin (\phi_1+2 \phi_2)+3 \sqrt{14} \sin (2 \phi_1+\phi_2)\right.\\
    &-9 \cos (\phi_1-\phi_2)-\left(50 \sqrt{2}+21\right) \cos (\phi_1+\phi_2) -3 \sqrt{2} \cos (2 \phi_1+\phi_2)+41 \sqrt{2} \cos (\phi_1+2 \phi_2)\\
    &+\left.\sqrt{7 \left(300 \sqrt{2}+697\right)} \sin (\phi_2)+25 \cos (\phi_2))\right]~.
\end{align*}}

In the table below we present estimated coherence ranges for $U_{49}$ and $U_{81}$ and compare them to permutations.

\begin{table}[h]
    \centering
    \begin{tabular}{|c|c|c|c|c|c|c|}
    \hline
         $\alpha$ & \multicolumn{2}{c|}{0} & \multicolumn{2}{c|}{2} & \multicolumn{2}{c|}{$\infty$} \\\hline
         \multicolumn{1}{c|}{} & min $S_0$ & max $S_0$ & min $S_2$ & max $S_2$ & min $S_{\infty}$ & max $S_{\infty}$ \\  \cline{2-7} \hline
         $P_{49}$ & 1 & 49 & $1/49$ & 1 & $1/7$ & 1 \\
         $U_{49}$ & $31/7$  & 49 & 0.042 & $115/343$ & 0.27... & $\frac{7+6\sqrt{14}}{49}$ \\ \hline \hline 
         $P_{81}$ & 1 & 81 &  $1/81$ & 1 & $1/9$ & 1 \\
         $U_{81}$ & $7/3$ & 81 & $5/729$ &  $5/9$ & $1/9$ & $\frac{3+2\sqrt{3}}{9}$ \\ \hline
    \end{tabular}
    \caption{Comparison of coherence ranges of $S_0$ (the average number of non-zero entries of each row of a matrix),
    $S_2$ (related to the average purity of such a vector) and  $S_{\infty}$ (mean value of the largest entry of each vector), for 2-unitary permutation matrices $P_{d^2}$ and new construction of 2-unitary matrices $U_{49}$ and $U_{81}$. To simplify the expressions we fixed the parameters of $U_{49}$ by setting $\phi_1 = \phi_2 = 0$, and for $U_{81}$ we focused on the most incoherent case with $a_i = b_i = c_i = \frac{1}{\sqrt{3}}$, \mbox{$\theta_i = \phi_i = \frac{2\pi}{3})$}. }
    \label{tab:my_label}
\end{table}

\section{Calculation of entangling power for tristochastic channels}\label{app:q_tri_call}

In this Appendix, we explicitly derive the results discussed in Section \ref{sec:tristoch}.
First, let us focus on the entangling power $e_p$ and gate typicality $g_t$ for unitary matrix corresponding to convolutional channel $\Phi_A$: $U_{ki,lj} = A_{klj} (a_{k,l})_i$. Taking into account that $A_{klj}$ is a permutation tensor and $\{|a_{k,l}\ra\}$ is a $l^{\text{'th}}$ basis vector from the $k^{\text{'th}}$ basis, one can calculate that:

\begin{equation}
 \mathcal{E}\qty(\ket{U}) = 1 - \frac{1}{d^4} \sum_{k,l,k',l',j} A_{klj} A_{k'l'j} |\la a_{k,l}| a_{k',l'}\ra|^2 ~,\mathcal{E}\qty(\ket{US}) = 1 - \frac{1}{d^4} \sum_{k,k',l} |\la a_{k,l}| a_{k',l}\ra|^2~.
\end{equation}

Therefore we immediately get the following bounds
\begin{equation*}
1 - \frac{1}{d} \leq \mathcal{E}\qty(\ket{U}) \leq 1 - \frac{1}{d^2} ~,~~~~~ 1 - \frac{1}{d} \leq \mathcal{E}\qty(\ket{US}) \leq 1 - \frac{1}{d^2}~.
\end{equation*}
Using the above, we establish a general bound for the entangling power and gate typicality for the convolutional channels
\begin{equation}
\frac{d-1}{d+1} \leq e_p(U) \leq 1 ~,~~~~~ \frac{1}{2} - \frac{1}{2d + 2} \leq g_t(U) \leq \frac{1}{2} + \frac{1}{2d + 2}.
\end{equation}

However, after a closer look, one can obtain a tighter bound for entangling power. Let us consider the sum $\mathcal{E}\qty(\ket{U}) + \mathcal{E}\qty(\ket{US})$,
\begin{equation}
\label{eu+es}
\begin{aligned}
\mathcal{E}\qty(\ket{U}) + \mathcal{E}\qty(\ket{US}) =  2 & - \frac{1}{d^4} \sum_{k,l,k',l',j} A_{klj} A_{k'l'j} |\la a_{k,l}| a_{k',l'}\ra|^2  - \frac{1}{d^4} \sum_{k,k',l} |\la a_{k,l}| a_{k',l}\ra|^2 \\
 =  2 & - \frac{1}{d^4} \qty(\sum_{k,l,l',j} A_{klj} A_{kl'j} |\la a_{k,l}| a_{k,l'}\ra|^2 + \sum_{k\neq k',l,l',j} A_{klj} A_{k'l'j} |\la a_{k,l}| a_{k',l'}\ra|^2)   \\
& - \frac{1}{d^4} \qty(\sum_{k,l} |\la a_{k,l}| a_{k,l}\ra|^2 +  \sum_{k \neq k',l} |\la a_{k,l}| a_{k',l}\ra|^2)~.
\end{aligned}
\end{equation}
In order to understand this expression better we need to manipulate the indices to our advantage.
In the first and second sums, we see that for each value of $k$ and $l$ there exists only one value of $j$ such that $A_{klj}$ is nonzero, hence only these components contribute to the sums. Moreover, in the first sum both $A_{klj}$ and $A_{kl'j}$ are simultaneously nonzero only if $l' = l$, because $A_{klj}$ is a permutation tensor. By the similar argument let us define $\sigma_{kl}(k')$ as the only value of $l'$ such that the product $A_{k'l'j} \, A_{klj}$ is nonzero. After all this renaming above calculations can be summarized as  

\begin{equation*}
\begin{aligned}
\mathcal{E}\qty(\ket{U}) + \mathcal{E}\qty(\ket{US}) =
 2& - \frac{1}{d^4} \sum_{k,l} |\la a_{k,l}| a_{k,l}\ra|^2 - \frac{1}{d^4} \sum_{k\neq k',l} |\la a_{k,l}| a_{k',\sigma_{kl}(k')}\ra|^2 \\
 &- \frac{1}{d^4} \sum_{k,l} |\la a_{k,l}| a_{k,l}\ra|^2 - \frac{1}{d^4} \sum_{k \neq k',l} |\la a_{k,l}| a_{k',l}\ra|^2 = \\
=2 & - 2\frac{1}{d^2} - \frac{1}{d^4} \sum_{k\neq k',l} |\la a_{k,l}| a_{k',\sigma_{kl}(k')}\ra|^2  +|\la a_{k,l}| a_{k',l}\ra|^2 \geq \\
\geq2 & - \frac{2}{d^2} - \frac{d^2(d-1)}{d^4} = 2- \frac{d+1}{d^2}~,
\end{aligned}
\end{equation*}   
\end{widetext}
where in the last line we used the fact that $|\la a_{k,l}| a_{k',\sigma_{kl}(k')}\ra|^2  +|\la a_{k,l}| a_{k',l}\ra|^2$ is a sum of squared amplitudes of two coefficients of vector $|a_{kl}\ra$ in the basis $\{a_{k',l'}\}_{l' = 1}^d$,  so by normalization it must be smaller than $1$. Inserting obtained bound for $\mathcal{E}\qty(\ket{U}) + \mathcal{E}\qty(\ket{US})$ into the formula for entangling power one finds:

\begin{equation}
1 - \frac{1}{d+1} \leq e_{p}(U)
\end{equation}

If  all the bases $\{a_{k,l}\}_{l = 1}^d$ are mutually unbiased~\cite{mub_paper}, then $|\la a_{k,l}|a_{k',l'}\ra|^2 = \frac{1}{d}$ for $k \neq k'$ and $|\la a_{k,l}|a_{k,l'}\ra|^2 = \delta_{l,l'}$ as in equation \eqref{MU_def}, we obtain a unitary $U = U_{\text{MUB}}$ with unbiased basis in \eqref{U_form1}. This, in turn, lets us explicitly calculate:
\begin{equation*}
e_p(U_{\text{MUB}}) = 1 - \frac{2}{d^2 + d} ~,~~~~~ g_t(U_{\text{MUB}}) = \frac{1}{2} 
\end{equation*}
No matter which permutation tensor $A_{klj}$ we start with.

Next, let's discuss the average values of entangling power $e_p$ and gate typicality $g_t$.

\begin{theo}
The average value of entangling power $e_p$ and gate typicality $g_t$ of unitaries corresponding to convolutional channels is the same as for $U_{\text{MUB}}$, presented above.
\end{theo}

\begin{proof}
Let us start by rewriting each basis $\{| a_{kl}\ra\}_{l = 1}^d$ as a unitary matrix $U_k$, 
\begin{equation*}
    \left[U_k\right]_{li} = (a_{kl})_i~.
\end{equation*}
Hence average over all bases can be rephrased as the integration of entangling power $e_p$ and gate typicality $g_t$ over $U(d)^{\otimes d}$ with Haar measures. Moreover, both entangling power $e_p$ and gate typicality $g_t$ are linear combinations of expressions of the form $|\la a_{kl}|a_{k'l'}\ra|^2 $, so one might change the order of integration and summation and focus only on the following integral:
\begin{equation*}
\begin{aligned}
& \int_{SU(d)} dU_1 \cdots \int_{SU(d)} dU_d \; |\la a_{kl}|a_{k'l'}\ra|^2 =\\
&=\int_{SU(d)} dU_{k'} \int_{SU(d)}dU_k | [U_k U_{k'}^\dagger]_{ll'} |^2 = \\
&=\int_{SU(d)} dU_{k'} \int_{SU(d)}d (U_k U_k')| [U_k]_{ll'} |^2 =  \\
&=\int_{SU(d)} dU_{k'} \int_{SU(d)}d U_k | [U_k]_{ll'} |^2  = \frac{1}{d}
\end{aligned}
\end{equation*}
for $k \neq k'$, where we used the fact that $\int_{SU(d)} dU = 1$ and unitary invariants of Haar measure. For $k = k'$ one gets
\begin{equation*}
\begin{aligned}
& \int_{SU(d)} dU_1 \cdots \int_{SU(d)} dU_d \; |\la a_{kl}|a_{kl'}\ra|^2 =\\
& = \int_{SU(d)}dU_k | [U_k U_{k}^\dagger]_{ll'} |^2 \!=\! \int_{SU(d)}dU_k | [U_k U_{k}^\dagger]_{ll'} |^2 \delta_{ll'} = \\
& = \delta_{ll'}~.
\end{aligned}
\end{equation*}

Since the average value of $|\la a_{kl}|a_{k'l'}\ra|^2$ over all possible bases is the same as for MUB's, the average value of entangling power $e_p$ and gate typicality $g_t$ is the same as for $U_{\text{MUB}}$'s.
\end{proof}

For comparison note that the average over the entire unitary group $U(d^2)$ readers 
\begin{equation}
\label{standard_average_ep}
\la e_p(U_{d^2})\ra_{CUE} = \frac{(d-1)^2}{d^2 +1}~,~~ \la g_t(U_{d^2})\ra_{CUE} = \frac{1}{2}~.
\end{equation}

After becoming acquainted with the behaviour of convolutional channel we are ready to present the proof of the Theorem \ref{theo_ep_stoch}.

\begin{proof}
Assume that there exists some  unitary $U_{ki,lj} = A_{klj} (a_{kl})_i$ with maximal entangling power $e_p(U) = 1$. Since maximal entangling power translates to the maximal value of the sum $\mathcal{E}\qty(\ket{U}) + \mathcal{E}\qty(\ket{US})$, by equation \eqref{eu+es} (and discussion therein), it corresponds to

\begin{equation}
\label{max_ep_cond1}
0 = \sum_{k\neq k',l,l',j} A_{klj} A_{k'l'j} |\la a_{k,l}| a_{k',l'}\ra|^2,  
\end{equation}
\begin{equation}
\label{max_ep_cond2}
0 =  \sum_{k \neq k',l} |\la a_{k,l}| a_{k',l}\ra|^2.
\end{equation}

On the other hand, quantum tristochasticity of the channel $\Phi_A$ is equivalent to the condition that, 
\begin{equation*}
\Phi_A[\rho\otimes   \rho^*] = \Phi_A[\rho^*\otimes   \rho]  = \rho^*~,
\end{equation*}
for any $\rho$, where $\rho^*$ is a maximally mixed state~\cite{bistron2023tristochastic}.

This property, in turn, is equivalent to conditions \eqref{max_ep_cond1},\eqref{max_ep_cond2}. Which can be seen from examination of the off-diagonal values of $\rho^* = \Phi_U(\rho, \rho^*)$:
\begin{equation}
0  =\rho_{kk'}^* = \sum_{l, l',j} A_{klj} A_{k'l'j} |\la a_{k,l}| a_{k',l'}\ra|^2 \frac{\rho_{l l'} }{d}~,
\end{equation}
which is true for any $\rho$ if and only if all terms in  \eqref{max_ep_cond1} are equal to zero. This, in turn, is equivalent to their sum being equal to zero due to their nonnegativity.

By placing the maximally mixed state in the second argument one gets
\begin{equation}
\begin{aligned}
& 0  = \rho_{kk'}^*= \sum_{l,j,j'} A_{klj} A_{k'lj'} |\la a_{k,l}| a_{k',l}\ra|^2 \frac{\rho_{jj'}}{d} = \\
& = \sum_l |\la a_{k,l}| a_{k',l}\ra|^2  \frac{\rho_{j(k,l)j(k',l)}}{d}~,
\end{aligned}
\end{equation}
where $j(k,l)$ is such that $A_{k,l,j(k,l)} = 1$.
This is equivalent to the condition \eqref{max_ep_cond2} by the same token as above.
\end{proof}

\section{Orthogonal gates with large entangling power in dimension $6\times 6$}\label{app:orto6}

Although we did not find quhex bipartite unitary channels with $e_p = 1$, we found several solutions attaining the same value of entangling power as the current record~\cite{bruzda2022structured} for orthogonal channels: $e_p = \frac{208 + \sqrt{3}}{210} \approx 0.9987$.
Below we present the corresponding Latin square and a bases $\{|a_{k,l}\ra\}_{l = 1}^6$ giving such an exemplary channel by equation \eqref{U_form1}. The Latin square reads:

\begin{equation*}
\left(
\begin{array}{cccccc}
 1 & 2 & 3 & 4 & 5 & 6 \\
 2 & 1 & 4 & 3 & 6 & 5 \\
 5 & 6 & 1 & 2 & 3 & 4 \\
 6 & 5 & 2 & 1 & 4 & 3 \\
 3 & 4 & 6 & 5 & 1 & 2 \\
 4 & 3 & 5 & 6 & 2 & 1 \\
\end{array}
\right)    
\end{equation*}
and the corresponding bases, given as orthogonal matrices, are
{\scriptsize
\begin{equation*}
\begin{aligned}
& \left(
\begin{array}{cccccc}
 1 & \cdot & \cdot & \cdot & \cdot & \cdot \\
 \cdot & 1 & \cdot & \cdot & \cdot & \cdot \\
 \cdot & \cdot & 1 & \cdot & \cdot & \cdot \\
 \cdot & \cdot & \cdot & 1 & \cdot & \cdot \\
 \cdot & \cdot & \cdot & \cdot & 1 & \cdot \\
 \cdot & \cdot & \cdot & \cdot & \cdot & 1 \\
\end{array}
\right) &
& \left(
\begin{array}{cccccc}
 \cdot & \cdot & \cdot & 1 & \cdot & \cdot \\
 \cdot & \cdot & 1 & \cdot & \cdot & \cdot \\
 \cdot & \cdot & \cdot & \cdot & -a & -a \\
 \cdot & \cdot & \cdot & \cdot & a & -a \\
 1 & \cdot & \cdot & \cdot & \cdot & \cdot \\
 \cdot & 1 & \cdot & \cdot & \cdot & \cdot \\
\end{array}
\right) \\
& \left(
\begin{array}{cccccc}
 \cdot & 1 & \cdot & \cdot & \cdot & \cdot \\
 1 & \cdot & \cdot & \cdot & \cdot & \cdot \\
 \cdot & \cdot & \cdot & 1 & \cdot & \cdot \\
 \cdot & \cdot & 1 & \cdot & \cdot & \cdot \\
 \cdot & \cdot & \cdot & \cdot & -b & -b' \\
 \cdot & \cdot & \cdot & \cdot & b' & -b \\
\end{array}
\right)&
& \left(
\begin{array}{cccccc}
 \cdot & \cdot & \cdot & \cdot & c & -c' \\
 \cdot & \cdot & \cdot & \cdot & c' & c \\
 \cdot & 1 & \cdot & \cdot & \cdot & \cdot \\
 1 & \cdot & \cdot & \cdot & \cdot & \cdot \\
 \cdot & \cdot & 1 & \cdot & \cdot & \cdot \\
 \cdot & \cdot & \cdot & 1 & \cdot & \cdot \\
\end{array}
\right) \\
& \left(
\begin{array}{cccccc}
 \cdot & \cdot & 1 & \cdot & \cdot & \cdot \\
 \cdot & \cdot & \cdot & 1 & \cdot & \cdot \\
 \cdot & \cdot & \cdot & \cdot & c' & -c \\
 \cdot & \cdot & \cdot & \cdot & c & c' \\
 \cdot & 1 & \cdot & \cdot & \cdot & \cdot \\
 1 & \cdot & \cdot & \cdot & \cdot & \cdot \\
\end{array}
\right) &
& \left(
\begin{array}{cccccc}
 \cdot & \cdot & \cdot & \cdot & b' & b \\
 \cdot & \cdot & \cdot & \cdot & -b & b' \\
 1 & \cdot & \cdot & \cdot & \cdot & \cdot \\
 \cdot & 1 & \cdot & \cdot & \cdot & \cdot \\
 \cdot & \cdot & \cdot & 1 & \cdot & \cdot \\
 \cdot & \cdot & 1 & \cdot & \cdot & \cdot \\
\end{array}
\right)
\end{aligned}
\end{equation*}}
where $a = \frac{1}{\sqrt{2}}$, $b = \frac{\sqrt{3} -1}{2\sqrt{2}}$, $b' = \frac{\sqrt{3}+1}{2 \sqrt{2}}$, $c = \frac{1}{2}$, $c' = \frac{\sqrt{3}}{2}$.

In contrast to the previous results of orthogonal matrices close to $2$-unitary, obtained in~\cite{bruzda2022structured} by a numerical search, we propose a heuristic construction which leads to the explicit analytic result which can be extended into a continuous family of unitary matrices with the same entangling power close to unity.

\section{Entangling power in multipartite systems}\label{sec:multi_eq}

Entanglement becomes significantly more complex when shifting from bipartite to multipartite systems -- there is no unique entanglement measure, and to make matters worse, it requires more than a single one for full description. Thus, multipartite entanglement becomes more of a landscape and less of a line~\cite{bengtsson2016brief, szalay2015multipartite, guo2021genuine}. 
It follows naturally, that entangling power can also be defined in many ways. In our analysis we focus on the definition provided in ref.~\cite{linowski2019entangling}.

\begin{defi}
Entangling power $E$ for an $m-1$ partite unitary channel $U$ is defined as the entanglement generated by the map $U$ averaged over all possible separable states $\ket{\psi_{sep}} =  |\psi_1\ra \otimes |\psi_2\ra \otimes \cdots$, 
\begin{equation}
E(U) = \la  \mathcal{E}_m(U|\psi_{sep}\ra) \ra_{|\psi_{sep}\ra}.
\end{equation}
where the measure of multipartite entanglement $\mathcal{E}_m$ is taken to be the average of entanglements with respect to all possible bipartitions $p|q$ of the system.
\begin{equation}
\mathcal{E}_m(|\psi\ra) = \frac{1}{2^{m-2}-1} \sum_{p|q} \mathcal{E}(|\psi\ra_{p|q}) 
\end{equation}
\end{defi}

There exists a general analytical formula for multipartite entangling power, which we present below in the simplified form with dimensions of all components equal $d$.

\begin{theo}\label{multi_ep_def}\cite{linowski2019entangling}
The entangling power for an $m-1$ partite unitary channel $U$ can be calculated as an average of entangling powers $E_{p|q}(U)$ with respect to all bipartitions $p|q$, 
\begin{equation}
\label{ep_multi_def}
E(U) = \frac{1}{2^{m-2}-1}\sum_{p|q} E_{p|q}(U)~.
\end{equation}
Entangling power $E_{p|q}(U)$ of $U$ with respect to bipartition $p|q$ in turn, is given by
{\footnotesize
\begin{equation}
\label{Epq_def}
E_{p|q}(U) = 2\left(1\!-\!\left(\frac{d}{d+1}\right)^{m-1} \!\!\sum_{x|y} \Tr\left[\Tr_{p,x}[|U\ra\la U|]^2\right]\right),
\end{equation}}
where the sum $\sum_{x|y}$ is also taken with respect to all the bipartitions of $m-1$ subsystems.
\end{theo}

\begin{widetext}

\section{Coherification of multi stochastic permutation tensors}\label{app:milti_coch}

In this Appendix, we generalize the construction of optimal coherifications from~\cite{bistron2023tristochastic} for multi-stochastic permutation tensor. The obtained results let us establish multipartite convolutional channels as generalization of convolutional channels. Following ref. \cite{bistron2023tristochastic, Korzekwa_coch} a coherification is considered to be 'optimal' if the norm-2 coherence achieves its maximal value. This measure quantifies the average contribution of the non-diagonal entries of the dynamical matrix $D$ \cite{Korzekwa_coch},

\begin{equation}
\label{C2_def_eq}
C_2(\Phi_D) = \frac{1}{N^4}\left(\sum_{kmln}|(D_A)^{k\;n\;}_{\;l\;m}|^2 - \sum_{kmln}|(D_{A,\text{diag}})^{k\;n\;}_{\;l\;m}|^2\right) \\
= \frac{1}{N^4} \sum_{\mu} \lambda_{D,\mu}^2 - \frac{1}{N^4}\sum_{\nu} \lambda_{D_{T,\text{diag}},\nu}^2~.
\end{equation}
Here $\Phi_D$ is a given coherification of a tristochastic tensor $A$, $\lambda_{D,\mu}$ are eigenvalues of dynamical matrix $D_A$ and $\lambda_{D_{T,\text{diag}},\mu}$ are eigenvalues of the dynamical matrix of diagonal coherification $D_{A,\text{diag}}$ without any non-diagonal terms.

Let us consider the Kraus representation of the channel $\Phi_A$. The Kraus operators $\{K_k\}$ are, in this case rectangular matrices such that
\begin{equation*}
    \Phi_D[\rho_1 \otimes\cdots\otimes \rho_{m-1}] = \sum_k K_k (\rho_1 \otimes\cdots\otimes \rho_{m-1}) K_k^\dagger~~~,
\end{equation*}
so the connection between Kraus operators and the dynamical matrix is given as
\begin{equation*}
D_{\; i_1 \; \vb{I}}^{i_1' \; \vb{I}'} = 
 \sum_k (K_k)_{i_1}^{I} (\overline{K}_k)_{i_1'}^{I'} ~,
\end{equation*}
where $I$ denotes the combination of indices $i_2, \cdots, i_m$ in the following way: $I = i_2\; d^{m-1}+ i_3 d^{m-2} +\cdots+  i_{m-1}\; d+ i_m $ while  the multi index $\vb{I}$ is constructed as  $\vb{I} = i_2 \cdots i_m$.
  
To ensure that $\Phi_D$ is a coherification of $A$, we therefore demand that $\sum_k |(K_k)_{i_1}^{I}|^2 = A_{i_1 \vb{I}}$.
This implies that any Kraus operator $K_j$ can have nonzero entry $(K_j)_{i_1}^{I}$ if and only if $A_{i_1 I}$ is nonzero. Because for each $i_1 \cdots i_{m-1}$ there exist only one $i_m$ such that $A_{i_1 i_2 \cdots i_m} = 1$ we may enumerate those entries as $(a_{i_1;i_2,\cdots,i_{m-1}})$ and hereafter consider vectors $(a_{i_1;i_2,\cdots,i_{m-1}})_k$. Thus we will slightly abuse our notation, and use a multi-index $\vb{I}_{|}$ to denote  $i_2\cdots,i_{m-1}$. For example for the $4$-stochastic permutation tensor (4-dimensional hypercube):

\begin{equation}
\label{example_permutation_4tensor}
A = \left(\begin{matrix}
1 & 0 \\
0 & 1 \\
\end{matrix}\right.
\left|\begin{matrix}
0 & 1 \\
1 & 0 \\
\end{matrix}\right.
\Bigg{|}\Bigg{|}\begin{matrix}
0 & 1 \\
1 & 0 \\
\end{matrix}
\left|\begin{matrix}
1 & 0 \\
0 & 1 \\
\end{matrix}\right)~,
\end{equation}
we get Kraus operators of the form:
\begin{equation}
\label{eq41_v4}
K_k ={ \left(\begin{matrix}
(a_{(1;1,1)})_k & 0 & 0 & (a_{(1;1,2)})_k & 0 & (a_{(1;2,1)})_k & (a_{(1;2,2)})_k & 0 \\
0 & (a_{(2;1,1)})_k & (a_{(2;1,2)})_k & 0 & (a_{(2;2,1)})_k & 0 & 0 & (a_{(2;2,2)})_k \\
\end{matrix}\right)}~.
\end{equation}

Next we examine the condition $\sum_{k} K_k^\dagger K_k = \id$. Because in each column of each Kraus operator, there is only one nonzero parameter from the diagonal terms of $\sum_{k} K_K^\dagger K_k$  we obtain the condition $||a_{i_1;\vb{I}_{|}}||^2 = 1$.

Moreover because in $i_1^{\text{th}}$ row of each $K_k$ all coefficient $a_{i_1;\vb{I}_{|}}$ have first index the same and equal to $i_1$ (by the construction of these coefficients),  from the non-diagonal terms of $\sum_{i} K_k^\dagger K_k$ we get $\la a_{i_1; \vb{I}_{|}}|a_{i_1; \vb{I}_{|}'}\ra = 0$.  Thus for each value of $i_1$ the  vectors $\{|a_{i_1; \vb{I}_{|}}\rangle\}_{\vb{I}_{|} = 1}^{d^{m-2}}$ forms an orthonormal set, therefore the norm two coherification of $A$ would be maximal if this orthonormal set would spam the same space for any value of $i_1$. Hence the number of Kraus operators, which is equal to the dimension of vector space in which $| a_{i_1; \vb{I}_{|}}\ra$ lives, must be equal $d^{m-2}$. The norm two coherence $C_2$ of that coherification reads:
{\footnotesize
\begin{equation}
\begin{aligned}
 C_2(\Phi_D) & =  \frac{1}{d^{2(m-1)}}\Bigg( \sum_{i_1, i'_1, \vb{I}_{|}, \vb{I}_{|}' } |\la a_{i_1; \vb{I}_{|}}|a_{i_1'; \vb{I}_{|}'}\ra|^2 -   \sum_{i_1,\vb{I}_{|}} |\la a_{i_1; \vb{I}_{|}}|a_{i_1; \vb{I}_{|}}\ra|^2  \Bigg)=  \frac{1}{d^{2(m-1)}} \Bigg( \sum_{i_1,\vb{I}_{|}} \sum_{i'_1 \vb{I}_{|}' }  |\la a_{i_1; \vb{I}_{|}}|a_{i_1'; \vb{I}_{|}'}\ra|^2 -  \sum_{i_1,\vb{I}_{|}} 1 \Bigg) = \\
& = \frac{1}{d^{2(m-1)}} \Bigg( \sum_{i_1} \sum_{\vb{I}_{|}}  1 -  \sum_{i_1,\cdots, i_{m-1}} 1 \Bigg) = \frac{d - 1}{d^{(m-1)} },
\end{aligned}
\end{equation}}
where the third step comes from the decomposition of each (normalized) vector $|a_{i_1; \vb{I}_{|}}\ra$ in the basis $\{a_{i_1'; \vb{I}_{|}'}\ra\}_{\vb{I}_{|} = 1}^{d^{m-2}}$.

For such coherification of $m$-stochastic permutation tensor the construction of a quantum channel via a unitary channel and partial trace is analogical as in~\cite{bistron2023tristochastic}. Because we have $d^{m-2}$ Kraus operators $K_j$ we can construct unitary from them  in the following way
\begin{equation}
\begin{aligned}
\Phi_D[\rho_1\otimes\cdots\otimes \rho_{m-1}] & =  \Tr_{2,\cdots, m-1}\left[ U (\rho_1\otimes\cdots\otimes \rho_{m-1}) U^\dagger\right]~, \\
U & = \sum_{j_1,\cdots j_{m-2}} K_{d^{m-2}(j_1-1) + \cdots + j_{m-2} }  \otimes \left(|j_1\ra \otimes \cdots\otimes |j_{m-2}\ra \right)~,
\end{aligned}
\end{equation}
where we exchanged the index $k$, enumerating Kaus operators, by the set of indices $j_1,\cdots j_{m-2} \in \{1, \cdots,d\}$. Thus we obtained the desired structure of convolutional channels as unitaries followed by a partial trace.

Moreover, such unitaries once again have a structure of block diagonal unitary matrices with $d$ blocks: $B$, multiplied by some permutation matrix $P$ corresponding to underlying permutation tensor $A$ : $U = B P$ 
\begin{equation*}
U_{i_1,j_1,\;j_2,\;\cdots,\; j_{m-2}}^{\;i_2,\;i_3,\;i_4,\;\cdots\;,i_m\;} = A_{i_1,i_2,\cdots,i_m} (a_{(i_1;~i_2,\cdots,i_{m-1})})_{j_1, j_2,\cdots,j_{m-2}}~.
\end{equation*}
For example, for the permutation tensor \eqref{example_permutation_4tensor} we get:
\begin{equation}
\label{eq53_v4}
U = \left(\begin{matrix}
(a_{(1;1)})_{1} & 0 & 0 & (a_{(1;2)})_{1} & 0 & (a_{(1;3)})_{1} & (a_{(1;4)})_{1} & 0 \\
(a_{(1;1)})_{2} & 0 & 0 & (a_{(1;2)})_{2} & 0 & (a_{(1;3)})_{2} & (a_{(1;4)})_{2} & 0 \\
(a_{(1;1)})_{3} & 0 & 0 & (a_{(1;2)})_{3} & 0 & (a_{(1;3)})_{3} & (a_{(1;4)})_{3} & 0 \\
(a_{(1;1)})_{4} & 0 & 0 & (a_{(1;2)})_{4} & 0 & (a_{(1;3)})_{4} & (a_{(1;4)})_{4} & 0 \\
0 & (a_{(2;1)})_{1} & (a_{(2;2)})_{1} & 0 & (a_{(2;3)})_{1} & 0 & 0 & (a_{(2;3)})_{1} \\
0 & (a_{(2;1)})_{2} & (a_{(2;2)})_{2} & 0 & (a_{(2;3)})_{2} & 0 & 0 & (a_{(2;3)})_{2} \\
0 & (a_{(2;1)})_{3} & (a_{(2;2)})_{3} & 0 & (a_{(2;3)})_{3} & 0 & 0 & (a_{(2;3)})_{3} \\
0 & (a_{(2;1)})_{4} & (a_{(2;2)})_{4} & 0 & (a_{(2;3)})_{4} & 0 & 0 & (a_{(2;3)})_{4} \\
\end{matrix}\right)~~~,
\end{equation}
for which the block diagonal matrix yields:

\begin{equation}
\label{eq54_v4}
B = \left(\begin{matrix}
(a_{(1;1)})_{1} & (a_{(1;2)})_{1} & (a_{(1;3)})_{1} & (a_{(1;4)})_{1} & 0 & 0 & 0 & 0\\
(a_{(1;1)})_{2} & (a_{(1;2)})_{2} & (a_{(1;3)})_{2} & (a_{(1;4)})_{2} & 0 & 0 & 0 & 0\\
(a_{(1;1)})_{1} & (a_{(1;2)})_{1} & (a_{(1;3)})_{1} & (a_{(1;4)})_{1} & 0 & 0 & 0 & 0\\
(a_{(1;1)})_{1} & (a_{(1;2)})_{1} & (a_{(1;3)})_{1} & (a_{(1;4)})_{1} & 0 & 0 & 0 & 0\\
0 & 0 & 0 & 0 & (a_{(2;1)})_{1} & (a_{(2;2)})_{1} &  (a_{(2;3)})_{1} & (a_{(2;4)})_{1}\\
0 & 0 & 0 & 0 & (a_{(2;1)})_{2} & (a_{(2;2)})_{2} &  (a_{(2;3)})_{2} & (a_{(2;4)})_{2}\\
0 & 0 & 0 & 0 & (a_{(2;1)})_{3} & (a_{(2;2)})_{3} &  (a_{(2;3)})_{3} & (a_{(2;4)})_{3}\\
0 & 0 & 0 & 0 & (a_{(2;1)})_{4} & (a_{(2;2)})_{4} &  (a_{(2;3)})_{4} & (a_{(2;4)})_{4}\\
\end{matrix}\right)~,
\end{equation}
where we enumerated the parameters once again by $(a_{(i_1;\vb{I}_{|})})_{j} := (a_{(i_1; i_2\cdots i_{m-1})})_{j}$.

\section{Multistochasticity of unitary $U$ constructed form orthogonal Latin hypercubes}\label{app:milti_coch_multi}

In this Appendix, we present a proof of (quantum) multi-stochasticity of unitary operations, related to multipartite convolutional channels, constructed from orthogonal Latin cubes given in \eqref{multi_u_Latin}. However, before doing so, we must gently rewrite the unitary matrix of interest.

\begin{lem}
Any unitary channel of the form \eqref{multi_u_Latin}:
\begin{equation}
U = \sum_{i_2,\hdots,i_m}|L_{i_2,\cdots,i_m}^{(1)}, \cdots, L_{i_2,\cdots,i_m}^{(m-1)} \ra \la i_2, \cdots i_m|~,
\end{equation}
can be expressed by a new set on orthogonal Latin hypercubes $M^{(k)}$ as:
\begin{equation}
\label{U_id_2str}
U = \sum_{i_1,j_1,i_2,\hdots,i_{m-2}}| i_1, j_1, \cdots j_{m-2}\ra\la  M_{ i_1, j_1, \cdots j_{m-2}}^{(1)} \cdots M_{ i_1, j_1, \cdots j_{m-2}}^{(m-1)}|~,
\end{equation}
\end{lem}

\begin{proof}
By Theorem 5.3 from~\cite{LatinHypercubes} the problem of constructing $m-1$ orthogonal Latin hypercubes is equivalent to the construction of maximal distance separable (MDS) code of $d^{(m-1)}$ words of length $2(m-1)$ from an alphabet of length $d$, and distance between words equal $m$~\cite{LatinHypercubes}. Let us present this code as rows in the orthogonal array $OA$:

\begin{equation*}
OA(L^{(1)}, \cdots, L^{(m-1)})_{(i_2,\cdots, i_m)} =(i_2,\cdots, i_m, L_{i_2,\cdots,i_m}^{(1)}, \cdots, L_{i_2,\cdots,i_m}^{(m-1)} )~,
\end{equation*}
where the distance between any two rows, equal $m$, is understood as a number of coordinates in which two rows differ.

Notice that the distances between rows in $OA$ do not change if one moves the last $m-1$ columns to the front and put the rows in order to construct a new array $OA'$:

\begin{equation*}
OA'(L^{(1)}, \cdots, L^{(m-1)})_{(i_1, j_1,\cdots, j_{m-2})} = 
(i_1, j_1,\cdots, j_{m-2},  M_{ i_1, j_1, \cdots j_{m-2}}^{(1)} \cdots M_{ i_1, j_1, \cdots j_{m-2}}^{(m-1)})~.
\end{equation*}
where we renamed the indices:
\begin{equation*}
\begin{aligned}
& i_1 = L_{i_2,\cdots,i_m}^{(1)},& ~~~~  M_{ i_1, j_1, \cdots j_{m-2}}^{(1)} = i_2,\\
& j_1 = L_{i_2,\cdots,i_m}^{(2)}, & ~~~~  M_{ i_1, j_1, \cdots j_{m-2}}^{(2)} = i_3,\\
& \cdots &  \cdots~~~~~~~~~~~~~~~~~~  \\
\end{aligned}
\end{equation*} 

Thus we obtain a new maximal distance separable code which guarantees that the corresponding hypercubes $M$ are in fact orthogonal Latin hypercubes.
\end{proof}

\begin{theo}
\label{quantum_multi_thm}
The unitary channels defined by \eqref{multi_u_Latin} is a  (quantum) $m$-stochastic channel.
\end{theo}

\begin{proof}
By the above lemma, we may write the unitary matrix \eqref{multi_u_Latin} in a form \eqref{U_id_2str}, so the channel $\Phi_A$ acting on the set of input states, ${k'}^{\text{th}}$ of which is totally mixed $\rho^*$ gives:
\begin{equation}
\begin{aligned}
& \Phi_D[\rho_1 \otimes\cdots\otimes \rho^*\otimes \cdots\otimes \rho_{m-1}]_{i_1'}^{i_1}  = \text{Tr}_{2, \cdots (m-1)}[U (\rho_1 \otimes \cdots \otimes \rho^* \otimes \cdots \rho_{(m-1)} ) U^\dagger]_{i_1'}^{i_1} = \\
& = \sum_{j_1,\cdots, j_{(m-2)}} {\rho_1}_{M_{ i_1', j_1, \cdots j_{m-2}}^{(1)}}^{M_{ i_1, j_1, \cdots j_{m-2}}^{(1)}} \cdots
\frac{1}{d}  \delta_{M_{ i_1', j_1, \cdots j_{m-2}}^{(k)}}^{M_{ i_1, j_1, \cdots j_{m-2}}^{(k)}}
\cdots {\rho_{m-1}}_{M_{ i_1', j_1, \cdots j_{m-2}}^{(m-1)}}^{M_{ i_1, j_1, \cdots j_{m-2}}^{(m-1)}} = \\
&= \frac{1}{d} \delta_{i_1'}^{i_1}\sum_{j_1,\cdots, j_{(m-2)}} {\rho_1}_{M_{ i_1', j_1, \cdots j_{m-2}}^{(1)}}^{M_{ i_1, j_1, \cdots j_{m-2}}^{(1)}} \cdots
\cdots {\rho_{m-1}}_{M_{ i_1', j_1, \cdots j_{m-2}}^{(m-1)}}^{M_{ i_1, j_1, \cdots j_{m-2}}^{(m-1)}}  =  \frac{1}{d} \delta_{i_1'}^{i_1} ~.
\end{aligned}
\end{equation}
In the last step, we used the fact that each density matrix has the same values of upper and lower indices, and the sums run over all values of indices.
\end{proof}

\section{\rd{Maximal coherifications of stochastic tensors and connection to unitary matrices with maximal entangling power}}\label{app:max_coch}

\rd{In this Appendix we study the coherification of general stochastic tensors. We start by discussing maximal coherifications which leads us to a natural definition of uni-stochastic tensors analogous to uni-stochastic matrices. Then we reveal strong connection between the (quantum) multistochasticity and $m$-partite unitaries with maximal entangling power.
We base our consideration on purity of multipartite channels $\gamma(\Phi)$ \cite{Korzekwa_coch} which is defined in the same way as norm-2 coherence, but without extracting the ``diagonal'' part 
\begin{equation*}
\gamma(\Phi) = \frac{1}{d^{2(m-1)}} \sum_{\substack{i_1,\cdots i_m \\ i_1',\cdots i_{m}'}}|(D_{\Phi})_{\;i_1\;i_2\;\cdots}^{i_1' \; i_2\; \cdots}|^2  = \frac{1}{d^{2(m-1)}} \sum_k \lambda_k^2
\end{equation*}
where $\lambda_k$ are the eigenvalues of channel's dynamical matrix $D$ and factor $\frac{1}{d^{2(m-1)}}$ corresponds to normalization of $D$. }

\begin{lem}.
\rd{Purity of any quantum channel $\Phi: \Omega_{d^{m-1}}\to \Omega_d$ is upper bounded as $\gamma(\Phi) \leq \frac{1}{d^{m-2}}$. }
\end{lem}

\rd{\begin{proof}
Let $D_{i_1; i_2 \cdots i_m}^{i_1'; i_2' \cdots i_m'}$ be a dynamical matrix of the channel $\Phi: \Omega_{d^{m-1}} \to \Omega_d$ and $\{K_k\}$
the canonical set of Kraus operators \rd{obtained from eigenvectors of the corresponding Choi state}, so that
\begin{equation*}
    \sum_{i_1,\cdots, i_m} (K_k)_{i_1}^{i_2\cdots i_m} (\overline{K}_{k'})_{i_1}^{i_2\cdots i_m} \propto   \delta_{k,k'}.
\end{equation*}
Then the purity  of $\Phi$ is equal to
\begin{equation*}
\gamma(\Phi) = \frac{1}{d^{2(m-1)}} \lambda_k^2 =  \frac{1}{d^{2(m-1)}} \sum_k \norm{K_k}^4
\end{equation*}
where \rd{$||X|| = \sqrt{\Tr(X^\dagger X)}$ is the Hilbert-Schmidt norm and} the relation $\lambda_k = \norm{K_k}^2$ comes \rd{naturally} from the fact that \rd{the canonical set of} Kraus operators are the eigenvectors of Choi state rescaled by square root of corresponding eigenvalues. 
On the other hand, the \rd{trace-preserving property of the channel} $\Phi$ \rd{corresponds to the Kraus operators providing a resolution of identity,}
\begin{equation*}
    \sum_k K_k^\dagger K_k = \id_{d^{m-1}}~.
\end{equation*}
Because each Kraus operator is a $d \times d^{m-1}$ matrix, $K_k^\dagger K_k$ may have at most $d$ nonzero eigenvalues. Therefore, there must be at least $d^{m-2}$ Kraus operators. 
Moreover, because for each $k$ we have $K_k^\dagger K_k \leq \id$, each eigenvalue of $K_k^\dagger K_k$ must be smaller or equal than $1$ which gives
\begin{equation*}
\lambda_k = \norm{K_k}^2 = \Tr(K_k^\dagger K_k) \leq d \cdot 1 = d. 
\end{equation*} 
Thus we have 
at \rd{least} $d^{m-2}$ Kraus operators, with each norm squared less or equal $d$, the sum of norm squared equal $d^2$ due to normalization condition. The quantity of interest is the maximum sum of norms to the power four.
Thus the extremal case there are exactly $d^{m-2}$ Kraus operators, each with norm squared equal $d$, which gives:
\begin{equation*}
\gamma(\Phi) =  \frac{1}{d^{2(m-1)}} \sum_{k = 1}^{d^{m-2}} ||K_k ||^4 \leq   \frac{1}{d^{2(m-1)}} \sum_{k = 1}^{d^{m-2}} d^2 = \frac{1}{d^{m-2}}.
\end{equation*}
\end{proof}}

\rd{Since the minimal number of Kraus operator coincides with channels obtained from unitary evolution followed by a partial trace, it motivates us to define uni-stochastic tensors, as the ones for which the ``classical'' action can be obtained in exactly such a way. }

\rd{\begin{defi}\label{def:unimulti}
Tensor $A$ of rank $m$ and dimension $d$ is \rd{called} a uni-stochastic tensor, if there exists a unitary channel $U$ such that the action of $U$ on \rd{any set of} diagonal density matrices 
\rd{$\rho^{(i)} = \sum_j p_j^{(i)} \op{i}$}
coincides on diagonal with the action of $A$
\begin{equation}
\label{eq:unimulti}
\sum_{i_2,\cdots, i_m} A_{i_1,\cdots, i_m} p_{i_2}^{(1)} \cdots  p_{i_m}^{(m-1)} = \operatorname{diag}\left( \Tr_{2,\cdots m-1}\left[U \left(\rho^{(1)}\otimes\hdots\otimes\rho^{(m-1)} \right)U^\dagger \right]\right)~.
\end{equation}
where the trace is taken over all subsystems except the first
and $\operatorname{diag}(\rho)$ denotes the vector of diagonal entries of $\rho$.
We call each $U$ satisfying \eqref{eq:unimulti} a unitary coherification of $A$. Moreover if $A$ is also multistochastic, we call such $A$ uni-multistochastic. 
\end{defi}}

\rd{It is relatively easy to \rd{show} that unitary coherifications of a uni-stochastic tensors yeld maximal coherence:}

\rd{\begin{obs}
For any uni-stochastic tensor $A$, its unitary coherification $U$ gives a dynamical matrix $D_A$ with maximal purity $\gamma(\Phi_A) = \frac{1}{d^{m-2}}$, thus also a maximal two norm coherence.
\end{obs}}

\rd{\begin{proof}
The unitary coherification $U$ of uni-stochastic tensor $A$ gives a dynamical matrix
\begin{equation}\label{dynamical_from_uni}
\gamma(\Phi_A) =  D_{i_1, i_2 \cdots i_m}^{i_1', i_2' \cdots i_m'} = \sum_{j_1,\cdots j_{m-2}} U_{i_1, j_1,\cdots j_{m-2}}^{i_2,\cdots,i_m} \overline{U}_{i_1', j_1,\cdots j_{m-2}}^{i_2',\cdots,i_m'}
\end{equation}
Thus we can explicitly calculate its purity:
\begin{equation*}
\begin{aligned}
&\frac{1}{d^{2(m-1)}} \sum_{i_1,\cdots,i_m, i_1'\cdots i_m'} |D_{i_1, i_2 \cdots i_m}^{i_1', i_2' \cdots i_m'}|^2 =\frac{1}{d^{2(m-1)}}  \sum_{i_1,\cdots,i_m, i_1'\cdots i_m'} |\sum_{j_1,\cdots j_{m-2}}  U_{i_1, j_1,\cdots j_{m-2}}^{i_2,\cdots,i_m} \overline{U}_{i_1', j_1,\cdots j_{m-2}}^{i_2',\cdots,i_m'}|^2  = \\
& = \frac{1}{d^{2(m-1)}}  \sum_{i_1,\cdots,i_m, i_1'\cdots i_m'} \sum_{j_1,\cdots j_{m-2}} \sum_{j_1',\cdots j_{m-2}'}  U_{i_1, j_1,\cdots j_{m-2}}^{i_2,\cdots,i_m} \overline{U}_{i_1', j_1,\cdots j_{m-2}}^{i_2',\cdots,i_m'} \overline{U}_{i_1, j_1',\cdots j_{m-2}'}^{i_2,\cdots,i_m} U_{i_1', j_1',\cdots j_{m-2}'}^{i_2',\cdots,i_m'} = \\
& = \frac{1}{d^{2(m-1)}} \sum_{i_1,i_1'} \sum_{j_1,\cdots j_{m-2}} \sum_{j_1',\cdots j_{m-2}'}  \delta_{i_1, j_1,\cdots j_{m-2}}^{i_1, j_1',\cdots j_{m-2}'}    \delta_{i_1', j_1,\cdots j_{m-2}}^{i_1', j_1',\cdots j_{m-2}'} = \frac{1}{d^{2(m-1)}} d^2 d^{(m-2)} = \frac{1}{d^{m-2}}.
\end{aligned}
\end{equation*}
\end{proof}}

\rd{The recipe for coherification of a multi-stochastic permutation tensor can be interpreted as a proof, that each multi-stochastic permutation \rd{tensor} is uni-stochastic, in analogy to ordinary permutation matrix being uni-stochastic in the standard sense (without partial trace) \cite{Zyczkowski2003}. However we stress that similarly as for stochastic matrices, not all stochastic tensors are uni-stochastic.}

\rd{With all the observations in mind we are \rd{prepared} to present the connection between uni-stochastic tensors and quantum multistochasticity.}

\rd{\begin{lem}\label{lem:unicoh_ep}
Any uni-stochastic tensor $A$ whose unitary coherification $U$ has maximal entangling power is $m$-stochastic. Furthermore, its coherification is (quantum) $m$-stochastic as well.
\end{lem}}

\begin{proof}
\rd{Let us first prove the multistochasticity of $A$.
Let $U$ be the unitary coherification of $A$ with maximal entangling power. By Definition \ref{def:unimulti} the elements of $A$ are given by:
\begin{equation}\label{eq:A_U_rel}
A_{i_1,i_2,\cdots,i_m} = \sum_{j_1,\dots j_{m-2} } |U_{i_1,j_1,\dots j_{m-2}}^{i_2,\cdots,i_{m}}|^2.
\end{equation}
Thus\rd{, it is straightforward to see that}
\begin{equation*}
\sum_{i_1} A_{i_1,i_2,\cdots,i_m} = \sum_{i_1} \sum_{j_1,\dots j_{m-2} } |U_{i_1,j_1,\dots j_{m-2}}^{i_2,\cdots,i_{m}}|^2 = 1~.
\end{equation*}
For any other index $i_k$ let us define $U'$ as $U$ with indices $i_1$ and $i_k$ swapped. Since $U$ has maximal entangling power, $|U\ra$ is maximally entangled with respect to all possible bipartitions \cite{linowski2019entangling}. Therefore $U'$ is also a unitary matrix so we obtain
\begin{equation*}
\sum_{i_k} A_{i_1,i_2,\cdots,i_m} = \sum_{i_k} \sum_{j_1,\cdots j_{m-2} } |{U'}_{i_k,j_1,\cdots j_{m-2}}^{i_2,\cdots,i_{m}}  |^2 = 1
\end{equation*}}

\rd{To show the multi stochasticity of the corresponding channel $\Phi_A$, note that its action can be expressed \rd{by} $U$ acting on $m-1$ states, which connects its $m-1$ ``input indices'' to the states, and then partial trace of $m-2$ subsystems, leaving only one ``output index''.
On the other hand the action of complementary channels defined by dynamical matrix $D$ of $\Phi_A$, see Definition \ref{def:multo_perm_coch}, 
corresponds to the same procedure, but with swapped distinct output index $i_1$ with some input index $i_k$. Thus we may leverage unitarity of $U'$ one again to see that
\begin{equation*}
\Tr_{1,\cdots,k-1,k+1,\cdots,m-1 }\;[D(\underbrace{\rho_2^\top \otimes \cdots \otimes }_{k-1\text{ elements}}\mathbb{I} \otimes  \cdots \otimes \rho_{m}^\top) ] = \Tr_{2,\cdots,m-1}[U' (\rho_2^\top \otimes \cdots \otimes\rho_m^\top)(U')^\dagger]~,
\end{equation*}
is a well defined quantum channel as well.}
\end{proof}

\rd{As argued in the Section \ref{sec:multi_perm_coch},
the implication in the opposite direction fails for $m >3$ due to larger number of possible bi-partitions of $|U\ra$: $m(m-1)/2$, than ``directions'' in which dynamical matrix $D$ can be be applied: $m$. However for $m = 3$ we can show that the relation is reciprocal.}

\rd{\begin{theo}
For any uni-stochastic tensor $A$ of rank $m = 3$ its unitary coherification $U$ has maximal entangling power $e_p(U) = 1$ if and only if its coherification is quantum tristochastic.
\end{theo}}

\rd{
\begin{proof}
The implication in one direction \rd{follows from} Lemma \ref{lem:unicoh_ep}. To show the implication in opposite direction let us consider uni-stochastic tensor $A$ with unitary coherification 
\begin{equation*}
\Phi_A[\rho_1\otimes\rho_2] = \Tr_2[U (\rho_1 \otimes \rho_2)U^\dagger]~.
\end{equation*}
The conditions for quantum tristochasticity implies that \cite{bistron2023tristochastic}
\begin{equation*}
\Tr_2[U (\rho^* \otimes \rho)U^\dagger] = \rho^*~,~~~ \Tr_2[U (\rho \otimes \rho^*)U^\dagger] = \rho^*,
\end{equation*}
for any quantum state $\rho$, where $\rho^*\rd{ = \mathbb{I}/d}$ is a maximally mixed state. Thus for a unitary $U$ it means that
\begin{equation*}
\sum_{il} U_{ki,lj} \overline{U}_{k'i,lj'} = \delta_{k,k'} \delta_{j,j'} ~,~~~ \sum_{ij} U_{ki,lj} \overline{U}_{k'i,l'j} = \delta_{k,k'} \delta_{l,l'}
\end{equation*}
hence $U$ is a unitary under arbitrary permutation of indices, so it must have maximal entangling power \cite{jonnadula2017impact}.
\end{proof}

The above results can be interpreted from the opposite perspective. Given unitary $U$ with maximal entangling power, one can always construct a multistochastic tensor $A$ via \eqref{eq:A_U_rel}, for which $U$ provides a maximal coherification.
}

\end{widetext}

\bibliographystyle{quantum_abbr}
\bibliography{biblio}

\end{document}